\documentclass{article}

\usepackage{amsmath}
\usepackage{amsthm}
\usepackage{times}
\usepackage{bm}
\usepackage{caption}
\usepackage{subcaption}
\usepackage{graphicx}
\usepackage{amssymb}
\usepackage[page,title]{appendix}
\usepackage{mathrsfs}
\usepackage{comment}
\allowdisplaybreaks[4]
\usepackage{fancyhdr} 

\begin{document}

%

\title{Central limit theorems for network driven sampling }

\author{Xiao Li \\School of Mathematical Sciences\\ Peking University\\ \and Karl Rohe\\ Department of Statistics\\University of Wisconsin-Madison}
\date{}
\maketitle

\begin{abstract}
Respondent-Driven Sampling is a popular technique for sampling hidden populations. This paper models Respondent-Driven Sampling as a Markov process indexed by a tree.  Our main results show that the Volz-Heckathorn estimator is asymptotically normal below a critical threshold.  The key technical difficulties stem from (i) the dependence between samples and (ii) the tree structure which characterizes the dependence.
The theorems allow the growth rate of the tree to exceed one and  suggest that this growth rate should not be too large. To illustrate the usefulness of these results beyond their obvious use, an example shows that in certain cases the sample average is preferable to inverse probability weighting.  We provide a test statistic to distinguish between these two cases.

\end{abstract}

\section{Introduction}
Classical sampling requires a sampling frame, a list of individuals in the target population with a method to contact each individual (e.g. a phone number). For many populations, constructing a sampling frame is infeasible. Network driven sampling enables researchers to access populations of people, webpages, and proteins that are otherwise difficult to reach.  These techniques go by many names:  web crawling, Respondent-Driven Sampling, breadth-first search, snowball sampling, co-immunoprecipitation, and chromatin immunoprecipitation. In each application, the only way to reach the population of interest is by asking participants to refer friends.  

Respondent-Driven Sampling (RDS) serves as a motivating example for this paper.   The Centers for Disease Control, the World Health Organization, and the Joint United Nations Programme on HIV/AIDS have invested in RDS to reach marginalized and hard-to-reach populations \cite{heckathorn1997respondent, WHO}. Each individual $i$ in the population has a corresponding feature $y_i$ (e.g. $y_i\in \{0,1\}$ and $y_i = 1$ if $i$ is HIV+).  Using only the sampled individuals, we wish to make inferences about the average value of $y_i$ across the entire population, denoted as $\mu$ (e.g. the proportion of the population that is HIV+). Extensive previous statistical research has proposed various estimators of $\mu$ which are approximately unbiased based upon various types of models for an RDS sample \cite{salganik2004sampling, volz2008probability, gile2011improved}. We note that in the papers cited above (except \cite{gile2011improved}), RDS is assumed to sample with replacement. Previous research has also explored the variance of these estimators \cite{goel2009respondent, treevar}. This paper studies the asymptotic distribution of statistics related to these estimators. 

Results on asymptotic distributions for RDS are useful for two obvious reasons. First, they allow us to construct asymptotic confidence intervals for $\mu$. Second, they provide essential tools to test various statistical hypotheses. The only central limit theorem associated considered in the RDS literature studied the case when the tree indexed process reduces to a Markov chain \cite{goel2009respondent}; this presumes that each individual refers exactly one person. Previous research suggests that the number of referrals from each individual is fundamental in determining the variance of common estimators \cite{treevar}.  This paper establishes two central limit theorems in settings which allow for multiple referrals.

The main results apply to both the sample average and the Volz-Heckathorn estimator, which is an approximation of the inverse probability weighted estimator (cf Remark \ref{remark:transform}).  Because the inverse probability weighted (IPW) estimator  and its extensions are asymptotically unbiased, these estimators  are often preferred to the sample average.  However, sometimes survey weights are not needed and they only introduce additional variance to the estimator \cite{kenneth2016weights}.  This issue is particularly salient when sampling weights are highly heterogeneous, as is often the case in RDS.  Proposition \ref{prop:averageUnbiased} shows that if the outcomes $y_i$ are uncorrelated with the sampling weights, then the sample average is unbiased.  Theorem \ref{variancedifferenceTree} extends this result to RDS to show that the IPW estimator can have a larger variance than the sample average.   Taken together,  these results imply that the sample average can have a lower mean squared error (MSE) than the IPW estimator.  Section \ref{sec:example} introduces an estimator of the bias of the sample average.  The main results provide a path to test the null hypothesis that the bias is zero. This can be used to select between the sample average and the IPW estimator.  Section \ref{sec:addhealth} studies this routine with the AddHealth social network.

\section{Notation} \label{sec:notation}

\newtheorem{theorem}{Theorem}

\newtheorem{lemma}{Lemma}

\newtheorem{corollary}{Corollary}

\newtheorem{proposition}{Proposition}

\newtheorem{definition}{Definition}

\newtheorem{assumption}{Assumption}

\newtheorem{remark}{Remark}

Following \cite{goel2009respondent} and \cite{treevar}, the results below model the network sampling mechanism as a tree indexed Markov process on a graph.  There are many assumptions in this model which are incorrect in practice.  However, like the i.i.d assumption, it allows for tractable calculations.  In the simulations, we show that the theory derived from this model provides a good approximation for a more realistic sampling model.  \cite{lu2012sensitivity} studies the sensitivities of the estimators to this model.  

Let $G=(V,E)$ be a finite, undirected, and simple graph with vertex set $V=\{1,...,N\}$ and edge set $E$. $V$ contains the individuals in the population and $E$ describes how they are related to one another.  As discussed in the introduction,  $y:V\rightarrow \mathbb{R}$ is a fixed real-valued function on the state space $V$; these are the node features that are measured  on the sampled nodes. The target of RDS is to estimate $\mu= N^{-1}\sum_{i=1}^{N}y(i)$.  

If each sampled node referred exactly one friend, then the Markov sampling procedure would be a Markov chain.  Several classical central limit theorems exist for this model; see \cite{jones2004markov} for a review. The results herein allow for each sampled node to refer more than one node.  This is a Markov process indexed not by a chain, but rather by a tree.  Denote the referral tree as $\mathbb{T}$. Where the node set of $G$ indexes the population, the node set of $\mathbb{T}$ indexes the samples.  That is, we observe a subset of the individuals in $G$ with the sample $\{X_\tau\}_{\tau \in \mathbb{T}} \subset V$.  An edge $(\sigma, \tau)$ in the referral tree denotes that sampled individual $X_\sigma$ referred individual $X_\tau$ into the sample. Mathematically, $\mathbb{T}$ is a rooted tree--a connected graph with $n$ nodes, no cycles, and a vertex $0$ which indexes the seed node.  To simplify notation, $\sigma \in \mathbb{T}$ is used synonymously with $\sigma$ belonging to the vertex set of $\mathbb{T}$. 

For each non-root node  $\tau\in \mathbb{T}$, denote $p(\tau) \in \mathbb{T}$ as the parent of $\tau$ (i.e. the node one step closer to the root).  This paper presumes that $\{X_\tau\}_{\tau \in \mathbb{T}}$ is a tree-indexed random walk on $G$, which was a model introduced by \cite{benjamini1994markov}.  This model generalizes a Markov chain on $G$; each transition $X_{p(\tau)} \rightarrow X_\tau$ is an independent and identically distributed Markov transition with transition matrix $P$. Following \cite{benjamini1994markov}, we will call this process a $(\mathbb{T},P)$-walk on $G$. Unless stated otherwise, it will be presumed throughout that the root node of the random walk $X_0$ is initialized from the equilibrium distribution of $P$. It follows that $X_{\sigma}$ has distribution $\pi$ for all $\sigma \in \mathbb{T}$. 

Unless stated otherwise, the results in this paper allow for the transition matrix $P$ to be constructed from a weighted graph $G$. Let $w_{ij}$ be the weight of the edge $(i,j) \in E$; if $(i,j) \not \in E$, define $w_{ij} = 0$. If the graph is unweighted, then let $w_{ij} = 1$ for all $(i,j) \in E$. Define the degree of node $i$ as 
$deg(i) = \sum_j w_{ij}$. If the graph is unweighted, then $deg(i)$ is the number of connections to node $i$.  Throughout this paper, the graph is undirected.  So, $w_{ij} = w_{ji}$ for all pairs $i,j$.  
Given that $\{X_{p(\tau)} = i\}$, the probability of $\{X_\tau =j\}$ is proportional to $w_{ij}$;
\[P\left( X_\tau =j | X_{p(\tau)} = i \right) = \frac{w_{ij} }{ deg(i)}.\]
We use the term \textbf{simple random walk} for the Markov chain constructed on the unweighted graph (i.e. $w_{i,j} \in \{0,1\}$ for all $i,j$).  The simple random walk presumes that each participant selects a friend uniformly and independently at random from their list of friends. 

In order to estimate $\mu$, we observe $y(X_\tau)$ for all $\tau \in \mathbb{T}$.  Because $G$ is undirected, $P$ is reversible and has stationary distribution $\pi$ with $\pi_i \propto deg(i)$ for all $i \in G$; this fact is helpful for creating an asymptotically unbiased estimator for $\mu$, particularly under the simple random walk assumption \cite{volz2008probability}. 

\begin{remark} \label{remark:transform}
In general, the quantity of interest $\mu= N^{-1}\sum_{i=1}^{N}y(i)$ is not equal to $E_{\pi}(y)$. As such, the sample average of $y(X_\tau)$'s
is a biased estimator for $\mu$. With inverse probability weighting,  define a new function $y'(i)=y(i)(N\pi_i)^{-1}$ and the respective estimator
$$\hat{\mu}_{IPW}=\frac{1}{n}\sum_{\sigma \in \mathbb{T}}y'(X_{\sigma})=\frac{1}{n}\sum_{\sigma \in \mathbb{T}}\frac{y(X_{\sigma})}{N\pi_{X_{\sigma}}},$$ 
where $n=|\mathbb{T}|$ is the sample size. Then, $E_{\pi}(\hat\mu_{IPW})=E_{\pi}(y')=\mu$.  As such, the sample average of the $y'(X_\tau)$'s is an unbiased estimator of $\mu$.  Unfortunately, the values $\pi_i$ are unknown.  
In practice, RDS participants are asked various questions to measure how many friends they have in $G$.  Under the simple random walk assumption, $\pi_i$ is proportional to the number of friends of $i$; this result also requires that the edges in $G$ are undirected, something that will be presumed throughout the paper.  Therefore the Volz-Heckathorn estimator 

$$\hat{\mu}_{VH}=\sum_{\sigma \in \mathbb{T}}\frac{y(X_{\sigma})/deg(X_{\sigma})}{\sum_{\tau \in \mathbb{T}}1/{deg(X_{\tau})}}$$
is in essence a H\'ajek estimator based upon $deg(i)$ \cite{volz2008probability}.  Under the simple random walk assumption, this estimator provides an asymptotically unbiased estimator of $\mu$. 
\end{remark}

For each node $\tau\in \mathbb{T}$, let $|\tau|$ be the distance of the node from the root; this is also called the ``wave'' of $\tau$. For every pair of node $\sigma, \tau \in \mathbb{T}$, define $d(\sigma, \tau)$ to be the distance between $\sigma$ and $\tau$ on $\mathbb{T}$ (as a graph). For each non-leaf node $\sigma \in \mathbb{T}$, let $\eta(\sigma)$ be the number of offspring of $\sigma$.  A tree is said to be an $m$-tree of height $h$ if $\eta(\sigma)=m$ for all $\sigma\in\mathbb{T}$ with $|\sigma|<h$ and $\eta(\sigma) = 0$ for all $|\sigma| = h$.  Here, both $m$ and $h$ are a natural numbers (i.e. $m,h\in \mathbb{N}$).  $\mathbb{T}$ is said to be Galton-Watson if $\eta(\sigma)$ are i.i.d random variables in $\mathbb{N}$. While the theorems below only study 2-trees; the computational experiments in Section \ref{sec:sim} suggest that the conclusions of the analytical results are highly robust to replacing the 2-tree with a Galton-Watson tree.  

 \cite{levin2009markov} serves as this paper's key reference for Markov processes. Following the notation in that text, define $E_{\pi}(y)=\sum_{i=1}^{N}\pi_iy(i)$ and $var_{\pi}(y)=E_{\pi}(y-E_{\pi}(y))^2$ for the function $y$.   

There are two primary concerns about the model and estimator used in the main results below.   First, the Markov model allows for resampling.  Second, the results below only apply to $m$-trees, not more general trees.  The simulations in Section \ref{sec:sim} suggest that the analytic results continue to hold under a more realistic setting that addresses both of these concerns.

\section{Main Results} \label{sec:main}

The threshold $m<\lambda_2^{-2}$ was previously identified in \cite{treevar} as being a critical threshold for the design effect of network driven sampling; beyond this threshold, the variance of the standard estimator does not decay at the standard rate. In other words, 
$$var(\frac{1}{\sqrt{|\sigma\in \mathbb{T}:|\sigma|\leq h|}}\sum_{\sigma\in \mathbb{T}:|\sigma|\leq h}y(X_{\sigma}))\rightarrow \infty$$as $h\rightarrow \infty$. As such, using the traditional scaling, no central limit theorem holds above the critical threshold.  
Because of this, the theorems focus on the  case $m<\lambda_2^{-2}$. 
When $m> \lambda_2^{-2}$, the simulations in Section \ref{sec:sim} suggest that the central limit theorem does not hold for any scaling. 

Theorem \ref{thm1} is a central limit theorem for an estimator constructed from the tree-indexed Markov chain. The theorem holds for any function $y$, any reversible transition matrix with second largest eigenvalue satisfying $|\lambda_2|\ne1$, and any  $m<\lambda_2^{-2}$.

\begin{theorem}
\label{thm1}
Suppose that $P$ is a reversible transition matrix with respect to the equilibrium distribution $\pi$, and that the eigenvalues of $P$ are $1=\lambda_1 > |\lambda_2| \geq... \geq |\lambda_N|$. Without loss of generality, suppose that $E_{\pi}(y)=0$. Define
$$Y_i=\frac{1}{\sqrt{m^i}}\sum_{\tau \in \mathbb{T}: |\tau|=i} y(X_{\tau}).$$If $\mathbb{T}$ is an $m$-tree with $m<\lambda_2^{-2}$, then $$\frac{1}{\sqrt{h}}\sum_{i=1}^{h} Y_i \rightarrow N(0, {\sigma}^2_0)$$in distribution,  where ${\sigma}^2_0=var_{\pi}((\sqrt{m}P-I)^{-1}y)-var_{\pi}(P(\sqrt{m}P-I)^{-1}y).$
\end{theorem}

The sequence of random variables considered in Theorem 1 are not exactly sample averages, but a reweighted form of sample average. Samples in the same wave are equally weighted, while samples from different waves are not. The following theorem provides theoretical guarantee on the distribution of sample average for a specific class of transition matrix and node feature.

\begin{theorem}
\label{thm2}
Let $\mathbb{T}$ be a $2-$tree. Without loss of generality, suppose that $E_{\pi}(y)=0$. Define $\hat \mu_h=\frac{1}{\sqrt{2}^h}\sum_{\sigma\in \mathbb{T}, |\sigma|\leq h}y(X_{\sigma})$. Suppose that

(c1) $E(\hat \mu_h^{2k+1})=0$ for all $h,k \in \mathbb{N}$;

(c2) for any function $f$ on $V$ satisfying $E_{\pi}f=0$, $||Pf||_{\infty}\leq |\lambda_2|||f||_{\infty}$;

(c3)$ |\lambda_2|<\frac{1}{\sqrt{2}}$;

then $${\hat \mu}_h \rightarrow N(0, {\sigma}^2_0)$$in distribution for some $\sigma_0^2$.

\end{theorem}

\begin{remark}

Condition (c1) is a technical condition on the symmetry of $\hat \mu_h$ that is necessary in the proof. The following proposition provides a sufficient condition for (c1).

\begin{proposition}
Suppose that $y$ is symmetric, i.e. for any $i\in V$ there exists $j$ such that $y(j)=-y(i)$. If $p(u,v)=P(y(X_{\sigma})=v|y(X_{p(\sigma)})=u)$ is well-defined and $p(u,v)=p(-u,-v)$ for all $u,v\in y(V)$, then condition (c1) is satisfied.
\end{proposition}

\begin{proof}
Under the conditions of the proposition, the distribution of $\hat \mu_h$ is symmetric with respect to 0. Thus $E(\hat \mu_h^{2k+1})=0$ for all $h,k \in \mathbb{N}$. 
\end{proof}

Conditions (c2)-(c3) can be substituted by the following condition (c2'):

(c2') There exists $c<\frac{1}{\sqrt{2}}$ such that for any function $f$ on $V$ satisfying $E_{\pi}f=0$, $||Pf||_{\infty}\leq c||f||_{\infty}$.

Condition (c2') is weaker than (c2) and (c3) combined, but is stronger than (c3) alone. To see this, let $f$ be the eigenfunction of the second eigenvalue, and it follows that $ |\lambda_2|<\frac{1}{\sqrt{2}}.$ It can be easily seen that one necessary condition for (c2') is that $$\sum_j|P_{ij}-\pi_j|<\frac{1}{\sqrt{2}}$$ for all $i\in V$. In other words, all the rows of $P$ must be close to $\pi$. As previously discussed, condition (c3) is actually a necessary condition for the central limit theorem \cite{treevar}, in the sense that the variance of $\hat \mu_h$ tends to infinity if $ |\lambda_2|\geq\frac{1}{\sqrt{2}}.$
\end{remark}

For clarity in the exposition of the theorem and the proof, we have only proved the theorem for the 2-tree.  We believe that similar results are likely to hold for more general $m$-trees.

\subsection{Extension to the Volz-Heckathorn estimator}

When $P$ is restricted to be the transition matrix of the simple random walk on $G$,  the following corollary shows that Theorem \ref{thm2} can be extended to the Volz-Heckathorn estimator \cite{volz2008probability}.

Denote $\bar{d}=\frac{1}{N}\sum_{i \in V}deg(i)$ as the average node degree.  
Following Remark \ref{remark:transform},  the IPW estimator contains $1/(N\pi_i)$ which is equal to $\bar{d}/deg(i)$.  The Volz-Heckathorn estimator first estimates $\bar{d}$ with  the harmonic mean of the observed degrees.  Because this harmonic mean converges to $\bar{d}$ in probability, the following corollary applies Slutsky's Theorem to give a central limit theorem for the Volz-Heckathorn estimator.

\begin{corollary}
\label{cor:vh}
Let $\mathbb{T}$ be a 2-tree. Suppose in particular that $P$ is the transition matrix of the simple random walk on $G$. Define a new node feature $y'(i)=y(i)/deg(i)$. Without loss of generality, suppose that $E_{\pi}y'=0$ (this is not equivalent to $E_{\pi}y=0$). Define
$$\hat\mu_{h,VH}=\hat\mu_{h}\hat{\bar{d}}=\frac{1}{\sqrt{2^h}}\sum_{\sigma\in \mathbb{T}, |\sigma|\leq h}y'(X_{\sigma})\hat{\bar d},$$where $$\hat{\bar{d}}=\frac{2^{h+1}-1}{\sum_{\sigma\in \mathbb{T}, |\sigma|\leq h}1/deg(X_\tau)}.$$ If the new node feature $y'$ and the transition matrix $P$ satisfy conditions (c1)-(c3) in Theorem \ref{thm2}, then
$$\hat\mu_{h, VH} \rightarrow N(0, {\sigma}^2_{0,VH})$$in distribution for some $\sigma_{0,VH}^2$.
\end{corollary}

\subsection{Illustrating the conditions with a blockmodel}\label{sec:block}

Consider $G$ as coming from a blockmodel with two blocks \cite{lorrain1971structural}.  In this blockmodel, each node $i$ is given a label $z(i) \in \{1,2\}$ and every edge weight $w_{i,j} = B_{z(i), z(j)}$ for some symmetric $2\times 2$ matrix $B$.  Suppose that $y_i = y_j$ if $z(i) = z(j)$.  This  model was previously studied in \cite{goel2009respondent} and it serves as an approximation to the Stochastic Blockmodel. 

Given the structural equivalence of nodes within the same block, it is sufficient to study the transition matrix between blocks, $\mathscr{P} \in \mathbb{R}^{2 \times 2}$. If $\mathscr{P}$ is a symmetric matrix with entries $p_{11}=p_{22}=p$ and $p_{12}=p_{21}=1-p$ for some value $p$, then it can be easily verified that condition (c1) is satisfied.  
Moreover, if $2p-1<\frac{1}{\sqrt{2}}$, then conditions (c2) and (c3) are also satisfied. Our theorem asserts that the Volz-Heckathorn estimator converges to the normal distribution in this model. 

More generally, suppose that the nodes in the blockmodel for $G$ are equally balanced between $2K$ blocks with node features $\{y_1,-y_1,\dots,y_K,-y_K\}$ and that the transition matrix $p(u,v)=pI(u=v)+\frac{1-p}{2K-1}I(u\neq v)$. We can verify that all the conditions are satisfied as long as $\frac{1}{2K}<p<\frac{1}{2K}+\frac{1}{2\sqrt{2}}.$

\section{Comparing the variance of inverse probability weighting to the bias of the sample average} \label{sec:example}

An estimator with a small mean square error (i.e. $E(\hat \mu - \mu)^2$) has a small bias and a small variance. It is generally known that inverse probability weighting provides an unbiased estimator of $\mu$.  However, survey weights can also drastically inflate the variance of the estimator. This matter has been heavily studied by survey statisticians and substantial literature have devoted to the methodologies and issues regarding the use of sampling weights; see \cite{pfeffermann1996weights}, \cite{biemer2007weights}, \cite{valliant2013weights}, and \cite{kenneth2016weights} for a review.  To determine whether one should use sampling weights in RDS, this section gives a test statistic for the null hypothesis that the sample average is unbiased.  The results in the previous section suggest a confidence region for this test statistic. 

 Denote $n = |\mathbb{T}|$ as the number of samples. The next results compare $\hat \mu_{IPW}$ to the sample average 
$$\hat{\mu}=\frac{1}{n}\sum_{\tau \in \mathbb{T}}y(X_{\tau}).$$
Proposition \ref{variancedifference} and Theorem \ref{variancedifferenceTree} highlight the dangers of inverse probability weighting by showing that it can increase the variance.  
Proposition \ref{variancedifference} studies the simplified case where the samples are i.i.d from the stationary distribution.  Following the proposition, Theorem \ref{variancedifferenceTree}  studies the more relevant setting of the $(\mathbb{T},P)$-walk on $G$.  To simplify the statements of Proposition \ref{variancedifference} and Theorem \ref{variancedifferenceTree} and their proofs, the node features $y(i)$ are presumed to be random variables. This condition could be removed with further technical conditions on the moments of $y: V\rightarrow\mathbb{R}$ and its relationship to $\pi$. 

\begin{proposition}
\label{variancedifference}
Suppose that $X_1,...,X_n$ are sampled independently from the stationary distribution $\pi$ and that $y(1),..., y(N)$ are $N$ uncorrelated and identically distributed random variables with finite second moment $\mu_2$. Let $C_1=max_{1\leq i\leq N}N\pi_i$ and $var(\pi)=\frac{1}{N}\sum_{i=1}^{N}\pi_i^2-\frac{1}{N^2}$, then 
\begin{equation}
\label{differenceEst}
var(\hat{\mu}_{IPW})-var(\hat{\mu})\geq\mu_2N(\frac{N}{nC_1^2}-1)var(\pi).
\end{equation}
Thus, as long as $N>C_1^2n$, which can be easily satisfied in practice,
$$var(\hat{\mu}_{IPW})>var(\hat{\mu}).$$
\end{proposition}

\begin{theorem}
\label{variancedifferenceTree}
Suppose that $\{X_\tau:\tau \in \mathbb{T}\}$ is a sample from the $(\mathbb{T},P)$-walk on $G$, and that $y(1),..., y(N)$ are $N$ uncorrelated and identically distributed random variables with finite second moment $\mu_2$. Assume that there exist constants $C_1, C_2$ and $C_3$ (not the same constants as in Proposition \ref{variancedifference}) such that $C_1N\leq d_i\leq C_2N$ for all $i$ and $N^2var(\pi)>C_3$. Then
\begin{equation}
\label{differenceEstTree}
var(\hat{\mu}_{IPW})-var(\hat{\mu})\geq \mu_2(\frac{N^2C_1^2}{nC_2^2}var(\pi)-\frac{C_2}{NC_1^2}),
\end{equation}
and there exists $C$ independent of $n$ such that $var(\hat{\mu}_{IPW})>var(\hat{\mu})$ as long as $N>Cn$.

\end{theorem}

Proposition \ref{variancedifference} shows that as $var(\pi)$ increases, the difference between the variance of $\hat \mu_{IPW}$ and $\hat \mu$ becomes larger.  Similarly, Theorem \ref{variancedifferenceTree} also shows that the difference between the variance of $\hat \mu_{IPW}$ and $\hat \mu$ increases as  $var(\pi)$ increases, given that the relative upper and lower bound of $d$ (i.e. $C_1$ and $C_2$) remain fixed. 
Recall that when $P$ is the simple random walk, the probabilities $\pi_i$ are proportional to node degree. An extensive literature (e.g. \cite{strogatz2001exploring, clauset2009power}) has found  that  empirical networks have highly heterogeneous node degrees. 
As such, Equation \ref{differenceEst} shows that the variance of $\hat{\mu}_{IPW}$ can be dramatically greater than the variance of $\hat{\mu}$. Moreover, both Proposition \ref{variancedifference} and Theorem \ref{variancedifferenceTree} suggest that the variance difference $var(\hat{\mu}_{IPW})-var(\hat{\mu})$ can be considerable if we sample only a small proportion of the whole population. This problem is particularly salient when the population is large.

The bias-variance tradeoff presents a dilemma between inverse probability weighting and the sample average.  This bias can be estimated.
For every $i\in V$, define a new node feature,
\begin{equation}
\label{newf}
y'(i)=y(i)(1-\frac{1}{N\pi_{i}})=y(i)(1-\frac{\bar{d}}{deg(i)}).
\end{equation}

\begin{proposition} \label{prop:averageUnbiased}
The mean of the new node feature satisfies
$$E_{\pi}(y')=E_{\pi}(y)-\mu,$$
which is the true bias of the sample average. Therefore, under the null hypothesis,
$$H_0: E_{\pi}(y')=0,$$
the sample average is an unbiased estimator. If $\pi$ and $y$ are uncorrelated (i.e.$\sum_{i\in V}\pi_iy(i)=\frac{1}{N}\sum_{i\in V}y(i)$), then $H_0$ is satisfied and the sample average is unbiased. 
\end{proposition}

Under the conditions of Proposition \ref{variancedifference}, $\pi$ and $y$ satisfy the condition of Proposition \ref{prop:averageUnbiased} in expectation (i.e. they are uncorrelated in expectation).  Both of these conditions imply that the outcome is unrelated to the sampling weight (in some way).  Under such conditions, both estimators are unbiased.  If it is also true that $var(\pi)$ is large, then $\hat \mu$ has a smaller variance.  In this scenario,  $\hat \mu$ is preferable to $\hat \mu_{IPW}$ in terms of MSE.   However, if the sample average is biased, then we must compare the bias and variance of the two estimators.  Asymptotically, the variance of both estimators will vanish, while the bias stays constant.  So, for sufficiently large sample size, one should use $\hat \mu_{IPW}$.  For smaller sample sizes, the bias of the sample average could be small (relative to the difference in the variances).  In such settings, there will be a crossover point, a sample size at which $\hat \mu_{IPW}$ becomes preferable to $\hat \mu$.  To distinguish between these two cases, we want to test the null hypothesis that the sample average is unbiased, i.e. $E_{\pi}(y')=0$. Or, more generally, we want to provide a confidence region for the bias of the sample average.  


If $\bar{d}$ is unknown, as is generally the case, we can estimate $\bar{d}$ by the harmonic mean \cite{salganik2004sampling}
$$\hat{\bar{d}}=\frac{n}{\sum_{\tau \in \mathbb{T}}1/deg(X_\tau)}.$$
Substituting $\bar{d}$ for $\hat{\bar{d}}$ in Equation \ref{newf} yields the new node feature based on the Volz-Heckathorn estimator \cite{volz2008probability}
$$y'_{VH}(i)=y(i)(1-\frac{\hat{\bar{d}}}{deg(i)}).$$
Similarly, define 
\begin{equation}
\label{newf2}
\widehat{bias}=\frac{1}{n} \sum_{\sigma \in \mathbb{T}}y'_{VH}(X_{\sigma})=\frac{1}{n} \sum_{\sigma \in \mathbb{T}} y(X_\sigma)\left(1-\frac{\hat{\bar{d}}}{deg(X_\tau)}\right) = \hat \mu-\hat \mu_{VH},
\end{equation}
then $\widehat{bias} = \hat \mu-\hat \mu_{VH}$ is an asymptotically unbiased estimator for the bias of $\hat \mu$.  It serves as a test statistic for the null hypothesis $H_0: E_{\pi}(y')=0$.


The theorems above suggest that $\widehat{bias}$ converges to the normal distribution. The rejection region is then 
$$
W=\{|\widehat{bias}|>1.96\frac{\hat{\sigma_0}}{\sqrt{n}}\}.
$$
where $\hat{\sigma_0^2}$ is an estimate of the variance.

\section{Estimating the variance} \label{sec:varhat}
For some node feature $\tilde y$ (e.g. HIV status $y$ or the $y'$ in Equation \eqref{newf} that motivate $\widehat{bias}$), let $\tilde \mu$ denote the sample average.  Denote $\sigma_{\tilde \mu}^2$ as $Var_{\mathbb{T}, P}(\tilde \mu)$, where the subscript $\mathbb{T}, P$ denotes that the data is collected via a $(\mathbb{T},P)$-walk on $G$. This subsection studies a simple plug-in estimator for $\sigma_{\tilde \mu}^2$.  

The following function is essential to expressing $\sigma$ \cite{treevar}.

\begin{definition} \label{def:probgen}
Select two nodes $I, J$ uniformly at random from the tree $\mathbb{T}$.  Define the random variable $D = d(I,J)$ to be the graph distance in $\mathbb{T}$ between $I$ and $J$.  Define $\mathbb{G}$ as the probability generating function for $D$,
\[\mathbb{G}(z) = E( z^{D}).\]
\end{definition}
In practice, $\mathbb{T}$ is observed.  So, the function $\mathbb{G}$ can be computed.  In many studies there are multiple seed nodes.  In these cases, we suggest computing $d(I,J)$ on a tree which has an artificial root node that connects to all of the seeds; this root node could be imagined as an individual that is responsible for finding the seed nodes.  In this tree, two different seed nodes would be distance 2 apart.  

Denote the autocorrelation at lag 1 of $\tilde y(X_\tau)$ as
\[R = \frac{Cov(\tilde y(X_{p(\tau)}), \tilde y(X_\tau))}{var_\pi(\tilde y)}.\]
Both $Cov(\tilde y(X_{p(\tau)}), \tilde y(X_\tau))$ and $var_\pi$ can be estimated with plug-in quantities.  Because the data has been sampled proportional to $\pi$,  the plug-in quantity for $var_\pi$ should not explicitly adjust for $\pi$.  
\[\widehat{var_\pi(\tilde y)} = \frac{1}{n} \sum_{\tau \in \mathbb{T}} (\tilde y(X_\tau) - \hat \mu_{VH})^2.\]
Similarly, 
\[\widehat{Cov(\tilde y(X_{p(\tau)}), \tilde y(X_\tau))} = \frac{1}{n-1} \sum_{\tau \in \mathbb{T}  \setminus 0} (\tilde y(X_{p(\tau)}) - \hat \mu_{VH})(\tilde y(X_\tau) - \hat \mu_{VH}),\]
where $\{\mathbb{T}  \setminus 0\}$ contains all nodes except the root node $0$ (because $p(0)$ does not exist).  
Using these plug-in quantities, define $\hat R$.  Then, the estimator is 
\[\hat \sigma_{\tilde \mu}^2 = \mathbb{G}(\hat R) \widehat{var_\pi(\tilde y)}.\]

A popular bootstrap technique for estimating $\hat \sigma_{\tilde \mu}^2$ resamples $y(X_\tau)$ as a Markov process (i.e. in addition to $X_\tau$ being a Markov process, the bootstrap procedure also assumes that $y(X_\tau)$ is Markov) \cite{salg}.  This model is akin to the blockmodel with two blocks in Section \ref{sec:block}.  The following assumption is weaker than this assumption. 

\textbf{Assumption 1:} $\tilde y(i) = \mu + \sigma f(i)$, where $\mu, \sigma \in \mathbb{R}$ and  $f: V \rightarrow \mathbb{R}$ is an eigenfunction of $P$ with $var_\pi(f) = 1$.

\begin{proposition} \label{prop:equality}
Under Assumption 1,
\[\sigma_{\tilde \mu}^2 = \mathbb{G}(R) var_\pi(\tilde y).\]
\end{proposition}
While Assumption 1 is weaker than the previous assumption in \cite{salg}, the next proposition highlights the danger of this assumption.  It uses a different assumption which is a rather weak assumption. 

\textbf{Assumption 2:} $\mathbb{G}$ is convex on $[\lambda_{\min},1]$, where  $\lambda_{\min}$ is the smallest eigenvalue of $P$. 

Because $\mathbb{G}$ is a probability generating function, it is always convex on $[0,1]$.  As such, we on need to be worried about negative values.  Recall, that the central limit theorems above only hold when $|\lambda_{\min}| < 1/\sqrt{2} \approx .7$ (the smallest possible value for $\lambda_{\min}$ is $-1$).  Some simulated trees given in the appendix suggest that if $\mathbb{G}$ is not convex, it often fails in the neighborhood of $-1$.  As such, the assumption that $|\lambda_{\min}| < 1/\sqrt{2} \approx .7$ is likely to imply Assumption 2.  In practice, one observes the referral tree $\mathbb{T}$.  Thus, one can compute the second derivative of $\mathbb{G}$.  
Eigenvalues of $P$ close to negative one arise in antithetic sampling, where adjacent samples are dissimilar.  For example, if the population in $G$ was heterosexuals and edges in $G$ represent sexual contact, then men would only refer women and vice versa.  In this case, $\lambda_{\min}$ would be exactly $-1$.  While easily imagined, such settings are not current practice for RDS.  As such, large an negative values are uncommon; $\lambda_{\min}$ is likely close to zero.  

The following proposition follows from an application of Jensen's inequality. 
\begin{proposition} \label{prop:auto}
Under Assumption 2,
\[\sigma_{\tilde \mu}^2 \ge \mathbb{G}(R) var_\pi(\tilde y).\]
\end{proposition}

Because Assumption 2 is not very restrictive, the inequality in Proposition \ref{prop:auto} highlights the danger in breaking Assumption 1 (and thus the Markov model in \cite{salg});  breaking Assumption 1 will lead to $\hat \sigma_{\tilde \mu}^2$ underestimating the variance.



\subsection{A bias adjusted estimator for $\mu$}
The test above also allows us to derive a more robust estimator of $\mu$. Define 
$$\widehat{bias}=\frac{1}{n}\sum_{i=1}^{n}y'(X_{i}).$$
Then $\widehat{bias} = \hat \mu-\hat \mu_{IPW}$. Using the hypothesis test to choose between the sample average and inverse probability weighting is akin to hard thresholding the bias adjustment.  Define
$$\widehat{bias}_t = \left\{
\begin{array}{cc}
\widehat{bias} & \mbox{if reject $H_0$}  \\
0 & \mbox{otherwise}.\end{array}\right.$$
The final estimator of $\mu$ is then 
$\hat \mu_{BA}=\hat \mu - \widehat{bias}_t$ (BA for bias adjustment).  This estimator is explored in Section \ref{sec:addhealth}.

\section{Numerical results}

\subsection{Simulation} \label{sec:sim}
In this section we illustrate the theoretical results on simulated data. 
The simulations are performed on networks simulated from the Stochastic Blockmodel \cite{holland1983stochastic}. The four colors in Figure 1 correspond to four different networks that are simulated from four different Stochastic Blockmodels. Each of the four networks has N =5,000 nodes, equally balanced between group zero and group one. The probability of a connection between two nodes in different blocks is $r$ and the probability of connection between two nodes in the same block is $p$. To control the eigenvalues of the $5000\times5000$ transition matrix, consider the transition matrix between classes given by $\mathscr{P}=E(D)^{-1}E(A)$. The second eigenvalue of $\mathscr{P}$ is \cite{rohe2012sp}
$$ \lambda_2(\mathscr{P})=\frac{p-r}{p+r}$$ where expectations are under the Stochastic Blockmodel.  In our simulation, the second eigenvalue of the actual transition matrix is typically very close to $\lambda_2(\mathscr{P})$. We take $p+r=0.01$ in all four Stochastic Blockmodels so that the average degree is about 25. As such, $\lambda_2(\mathscr{P})$ is actually controlled by $p-r$.

For each of the four networks we carry out four different sampling designs. Let $\mathbb{T}$ be either a $2-$tree or a Galton-Watson tree with $E(\eta(\sigma))=2$. For the Galton-Watson tree, the distribution of $\eta(\sigma)$ is uniform on $\{1,2,3\}$. For each $\mathbb{T}$, we consider both with replacement sampling (i.e. the $(\mathbb{T},P)$-walk on $G$) and without replacement sampling (i.e. referrals are sampled uniformly from the friends that have not yet been sampled). Note that the conditions of Theorem \ref{thm2} may be violated when either the Galton-Watson tree or without-replacement sampling is used. We take the first 8 waves of $\mathbb{T}$ as our sample. As such, the sample size is roughly $N/10$.  For each social network and sampling design, we repeat the sampling process 2000 times and compute $\hat\mu=\frac{1}{n}\sum_{i=1}^{n}y(X_i)$ for each sample. The Quantile-Quantile (Q-Q) plot of $ \hat\mu$ is shown in the left panel of Figure \ref{fig:simulation}; note that the QQ plot centers and scales each distribution to have mean zero and standard deviation one. In addition, we repeat the above simulation for the Volz-Heckathorn estimator, and the QQ plot of $ \hat\mu_{VH}$ is shown in the right panel of Figure \ref{fig:simulation}.

\begin{figure}[h]
\centering
\begin{subfigure}[b]{0.45\textwidth}
\centering
\includegraphics[width=\textwidth]{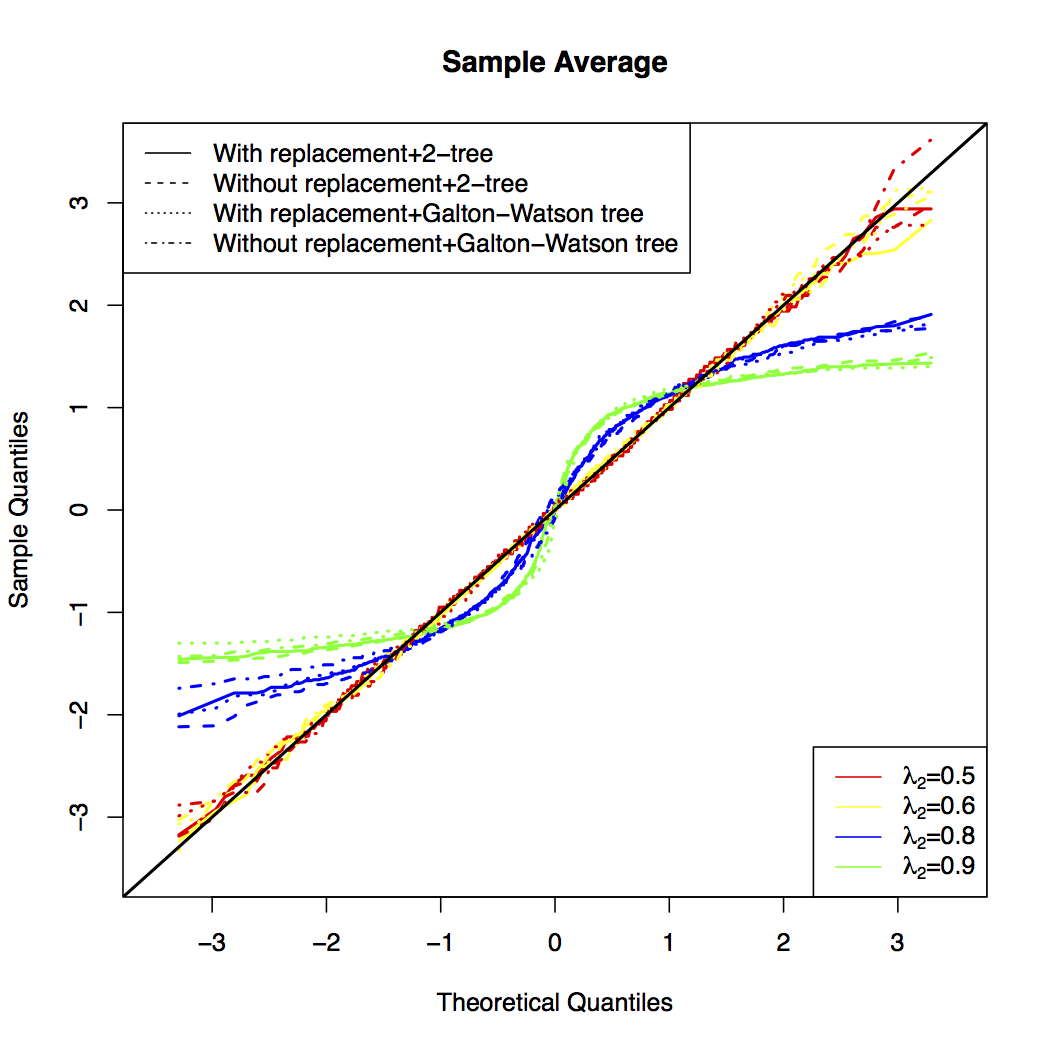}
\end{subfigure}
\begin{subfigure}[b]{0.45\textwidth}
\centering
\includegraphics[width=\textwidth]{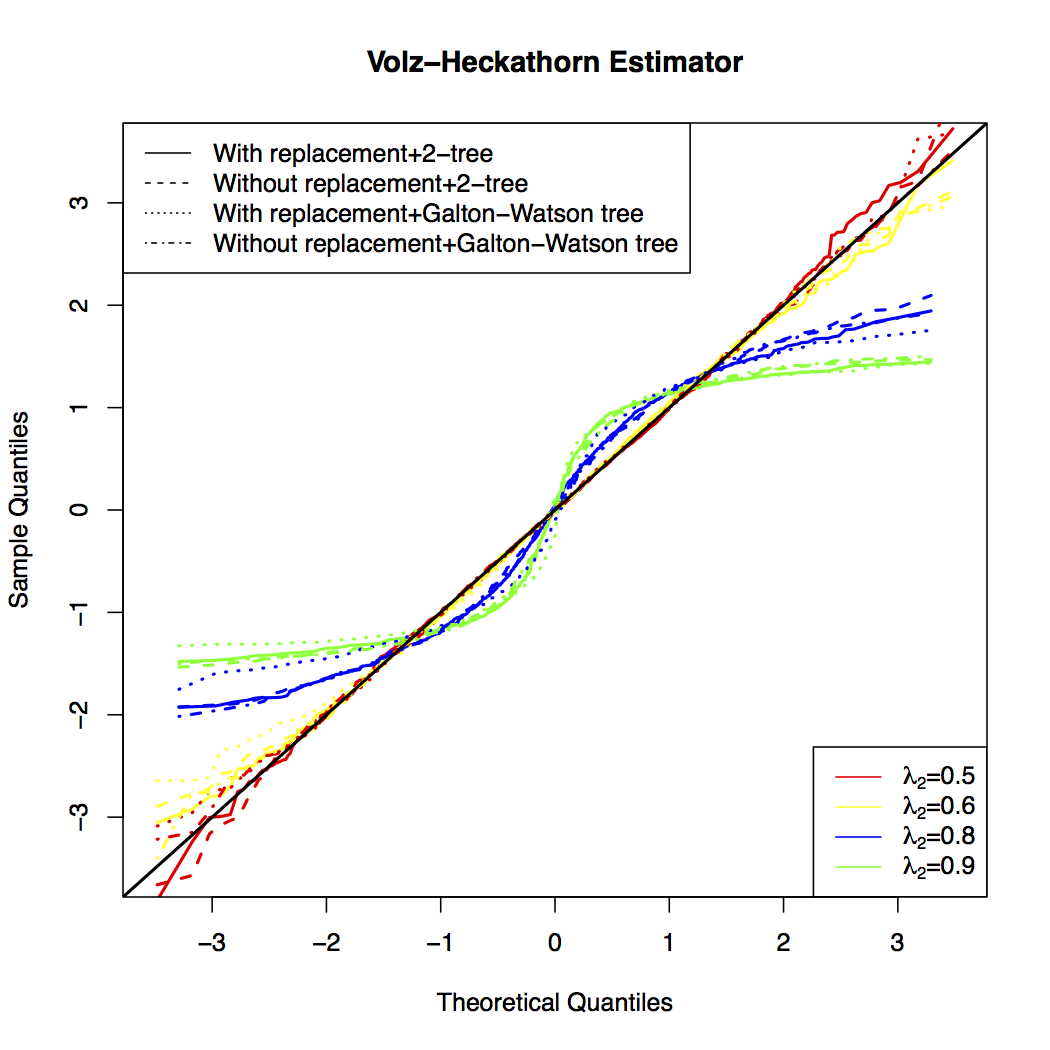}
\end{subfigure}
\caption{Q-Q plots of the sample average (left panel) and the Volz-Heckathorn estimator (right panel) for different social network and sampling designs. For each scenario we draw 2000 network driven samples of size $\approx$ 500 from a network containing 5,000 nodes. Here the threshold for $\lambda_2$ is $1/\sqrt{2} \approx 0.707$. For the two settings with $\lambda_2 < 1/\sqrt{2}$, the distributions appear normal.  However, for the two settings with $\lambda_2>1/\sqrt{2}$, the distributions do not appear normal.  Across all values of $\lambda_2$, there is no apparent difference between the four different designs (i.e. replacement sampling vs without replacement sampling and 2-tree vs Galton-Watson tree).}
\label{fig:simulation}
\end{figure}
\par

It is clear from Figure \ref{fig:simulation} that there are two patterns of distribution: when $\lambda_2<1/\sqrt{m}\approx 0.7$, i.e. $\lambda_2$=0.5 or 0.6, the Q-Q plots of $ \hat\mu$ and $ \hat\mu_{VH}$ approximately lie on the line $y=x$ for all sampling design;  when $\lambda_2>1/\sqrt{m}\approx 0.7$, i.e. $\lambda_2$=0.8 or 0.9, the Q-Q plot of $ \hat\mu$ and $ \hat\mu_{VH}$ departs from the line $y=x$. Taken together, Figure \ref{fig:simulation} suggests that the distribution of $ \hat\mu$ and $ \hat\mu_{VH}$ converges to Gaussian distribution if and only if $m<\lambda_2^{-2}$. Actually, the right panel of Figure \ref{fig:simulation} implies that there are two modes in the asymptotic distribution of $ \hat\mu$ and $ \hat\mu_{VH}$ when $m>\lambda_2^{-2}$. The relationship between the expectation of the offspring distribution and the second eigenvalue of the social network determines the asymptotic distribution of RDS estimators, regardless of the node feature, the particular structure of the tree or the way we handle replacement.

\subsection{Analysis of Adolescent Health data} \label{sec:addhealth}
To illustrate Theorem \ref{thm1} with the test statistic in Section 3, this section presents simulation results that use the Comm15 friendship network from the National Longitudinal Survey of Adolescent Health.  This simulation compares the MSE of $\hat \mu$  and $\hat \mu_{IPW}$ for two different node features $y$.  When $y$ is correlated with $\pi$, then $\hat \mu_{IPW}$ has a smaller MSE.  When $y$ is weakly correlated with $\pi$, then $\hat \mu$ has a smaller MSE.  For settings in which $\hat \mu_{IPW}$ clearly outperforms $\hat \mu$, the test statistic from Section 3 rejects the null hypothesis that the sample average is unbiased (i.e. $H_0: E_\pi \hat \mu = \mu$).

In the Comm15 network,  $N=1089$
students from two sister schools were asked to list up to 10 friends; these friends can be inside or outside of the school.  The students also supplied information including their gender, grade and race. 
The analysis below studies two node covariates: gender (0/1) and total nominations of friends (integer between 0 and 10).  Before the simulation, the network was symmetrized (i.e. consider the new adjacency matrix $\tilde A=A+A^T$), yielding a network with average node degree $\bar{d}=8.06$.  Because the students were only allowed to list up to 10 friends, the standard deviation of the degrees is $4.7$.  This is drastically smaller than typical social networks.  However, even in this setting, the variance of $\pi$ is sufficiently large to illustrate the advantages of the sample average.   

For both gender and the number of nominations, Table 1 displays (i) the correlation between $\pi$ and $y$, (ii) the bias of the sample average, and (iii) the crossover point.  Recall that the crossover point is the sample size at which the $\hat \mu_{IPW}$ has a smaller MSE than $\hat \mu$; this calculation is based upon the simulations described below.  The table shows that gender is weakly correlated with $\pi$.  As such, the sample average has a small bias and the crossover point is large.  Contrast this with the number of total nominations, which is highly correlated with $\pi$.  This makes the sample average clearly biased.  Because of this, it has a small crossover point.  These two examples illustrate a range of possibilities in terms of $cor(\pi, y)$.  
\par
\begin{table}
\centering
\label{truebias}
\begin{tabular}{l|lll}\\
&$cor(\pi,y)$&Bias of sample average &Crossover point (sample size)\\
\hline
Gender&0.11&0.037&$>500$\\
Total nominations&0.54&1.1&20\\
\end{tabular}
\caption{Gender has a weak correlation with the sample weights.  As such, $\hat \mu$ has a small bias.  The crossover point shows that if the samples were drawn independently from the distribution $\pi$, then $\hat \mu$ has a smaller MSE when $n<500$.  Total nominations has a larger correlation and thus a larger bias.  When the sample size exceeds twenty, $\hat \mu_{IPW}$ outperforms the sample average.}
\end{table}


%
%
%
%
%
%

%
Before estimating the crossover points shown in Table 1, we first study the hypothesis test $H_0: E \hat \mu = \mu$ for both gender and total nominations.  To provide a benchmark, this simulation compares RDS to independent sampling.  Let $P$ be the transition matrix of the simple random walk on the network. The second largest eigenvalue of $P$ is $\lambda_2=0.93$. Let $\mathbb{T}$ be a sample from the Galton-Watson process with $E\eta(\sigma)=1.1<\lambda_2^{-2}$. For a node covariate $y$ (gender or nominations), let $y'$ be the node feature defined in Eq.\ref{newf2} in Section 3. Recall that $E_{\pi}(y')$ is the bias of the sample average.  
 We consider the following methods of generating samples $Y_1, \dots, Y_n$ and computing or estimating the variance $\sigma^2$.  
\begin{enumerate}
\item $Y_1,..., Y_n$ is an independent and identically distributed sample from the normal distribution $N(E_{\pi}(y'), var_{\pi}(y'))$
and $\hat{\sigma}^2=var_{\pi}(y').$  Here the variance is known.  

\item $Y_i=y'(X_i)$ for all $i$, where $X_1,..., X_n$ is an independent and identically distributed sample from the equilibrium distribution of $P$, and $\hat{\sigma}^2=var_{\pi}(y').$ Here the variance is known. 

\item $Y_i=y'(X_i)$ for all $i$, where $X_1,..., X_n$ is a sample from the $(\mathbb{T},P)$-walk on $G$,
and $\hat{\sigma}^2$  is the true variance that is only computable in a simulation,  $var(\frac{1}{\sqrt{n}}\sum_{i=1}^{n}Y_i).$  

\item $Y_i=y'(X_i)$ for all $i$, where $X_1,..., X_n$ is a sample from the $(\mathbb{T},P)$-walk on $G$, and $\hat{\sigma}^2=\widehat{var_{\pi}(y')}\mathbb{G}(\hat R)$, the estimator discussed in Section \ref{sec:varhat}. 
\end{enumerate}

For a sample of size $n$, let the rejection region be$$W=\{|\frac{1}{\sqrt{n}\hat{\sigma}}\sum_{i=1}^{n}Y_i|>1.96\},$$with $Y_i$ and $\hat{\sigma}$ to be defined above.  Here,  the null hypothesis is rejected at $\alpha=0.05$.   For each scenario, Figure \ref{fig:power} plots the power of our test $pr(W)$ as a function of sample size. 
The power under scenario (1) can be calculated exactly and serves as a benchmark (black line). The power under scenario (2)-(4) is calculated by taking 1000 independent samples and counting the number of samples that fall in $W$ (red, blue, yellow lines respectively). 

Because scenario (4) underestimates the true variance, this technique is conservative in rejecting $H_0$ and adopting $\hat \mu_{IPW}$.  For gender, none of the scenarios quickly reject the null hypothesis. Compare this to the number of total nominations.  Here, $H_0$ is rejected even for small sample sizes. 

\begin{figure}
\centering
\begin{subfigure}[b]{0.45\textwidth}
\centering
\includegraphics[width=\textwidth]{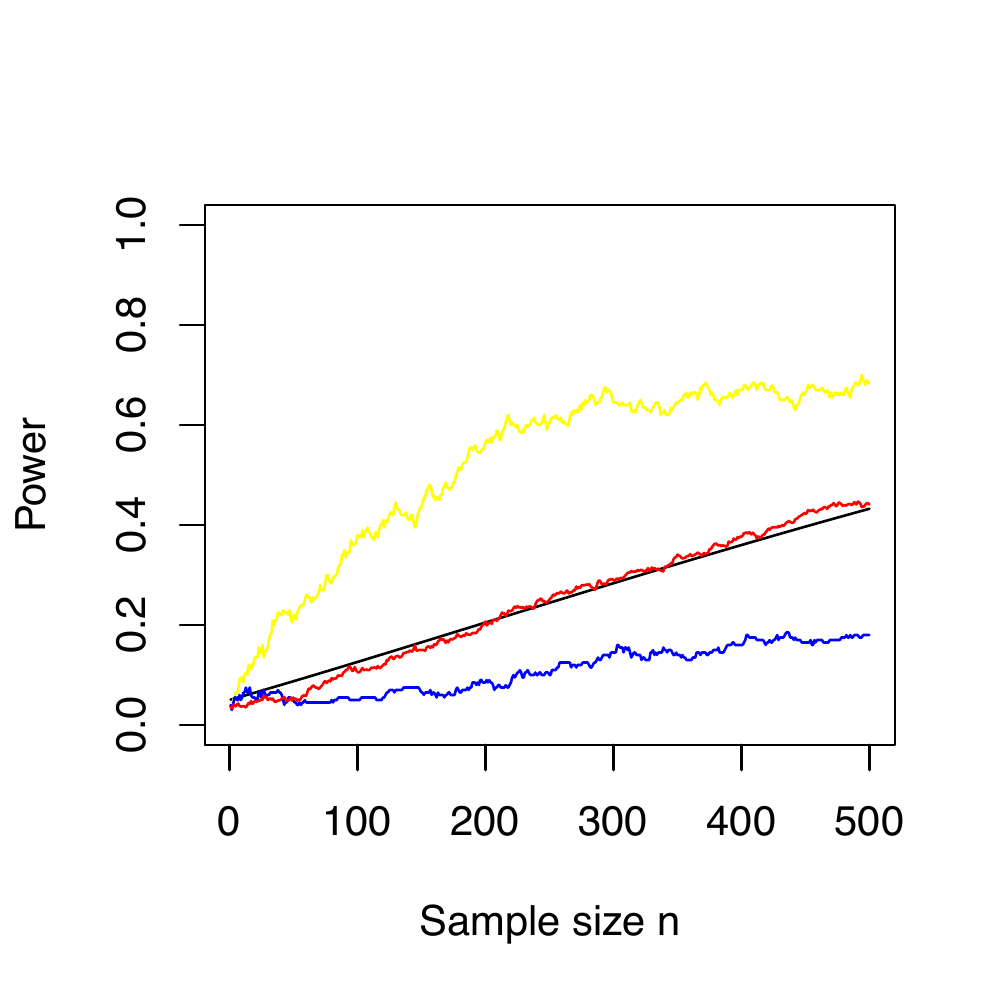}
\subcaption{Gender}
\end{subfigure}
\begin{subfigure}[b]{0.45\textwidth}
\centering
\includegraphics[width=\textwidth]{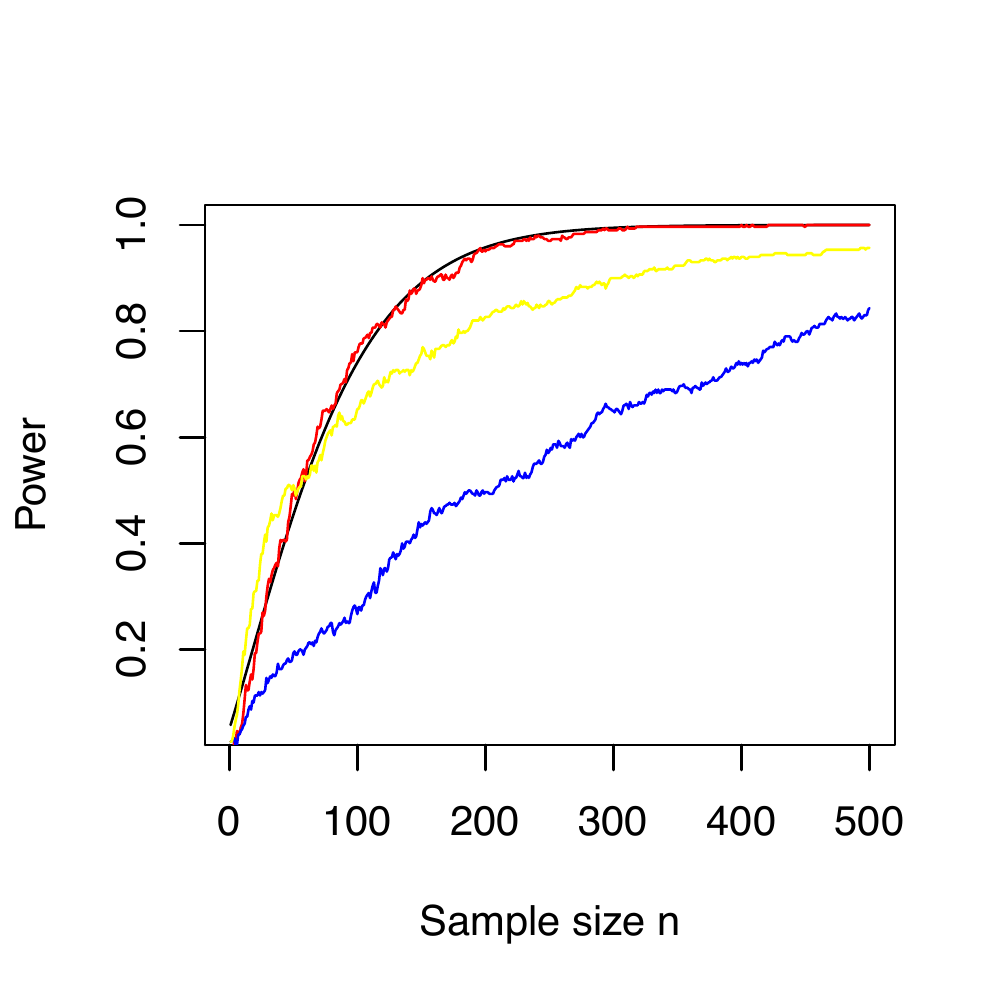}
\subcaption{Total nominations}
\end{subfigure}
\caption{Power of test as a function of sample size for two node features (a) gender and (b) total nominations of friends. The black, red, blue and yellow lines are the power under scenarios 1-4 respectively. }
\label{fig:power}
\end{figure}

The final figure plots the mean square error of $\hat \mu, \hat \mu_{IPW},$ and $\hat \mu_{BA}$; this last estimator is the bias adjusted estimator from Section \ref{sec:example}.  This simulation uses scenario  (4), the most realistic of the previous scenarios.  After drawing the sample, compute the following (for both gender and total nominations) 
\begin{enumerate}
\item The inverse probability weighted estimator $$\hat{\mu}_{IPW}=\frac{1}{n}\sum_{i=1}^{n}\frac{y(X_{i})}{N\pi_{X_{i}}}=\frac{1}{n}\sum_{i=1}^{n}\frac{y(X_{i})\bar{d}}{deg(X_i)}.$$
\item The sample average $\hat{\mu}=\frac{1}{n}\sum_{i=1}^{n}y(X_{i}).$
\item The bias adjusted estimator $$\hat{\mu}_{BA}= \left\{
\begin{array}{cc}
\hat{\mu}_{IPW} & \mbox{if $\{X_i\}_{1\leq i \leq n} \in W,$}  \\
\hat{\mu} & \mbox{if $\{X_i\}_{1\leq i \leq n} \notin W,$}
\end{array}
\right.$$
introduced in Section 3.
\end{enumerate}

Figure \ref{fig:mse} shows that for gender, the true bias of the sample average is small.  As such, the MSE of $\hat{\mu}$ (solid) is always smaller than that of $\hat{\mu}_{IPW}$ (dotted) (for sample sizes  $n<500$). For total nominations, the bias is much larger.  So, when $n>20$, the MSE of  $\hat{\mu}$ is larger than the MSE of $\hat \mu_{IPW}$. The MSE of the bias adjusted estimator $\hat{\mu}_{BA}$ (longdash) lies between that of $\hat{\mu}$ and $\hat{\mu}_{IPW}$.  In particular, when $\hat \mu_{IPW}$ drastically outperforms $\hat \mu$ (i.e. on the right of panel b), the null hypothesis is typically rejected and $\hat \mu_{BA}$ performs similarly to $\hat{\mu}_{IPW}$.

\begin{figure}
\centering
\begin{subfigure}[b]{0.45\textwidth}
\centering
\includegraphics[width=\textwidth]{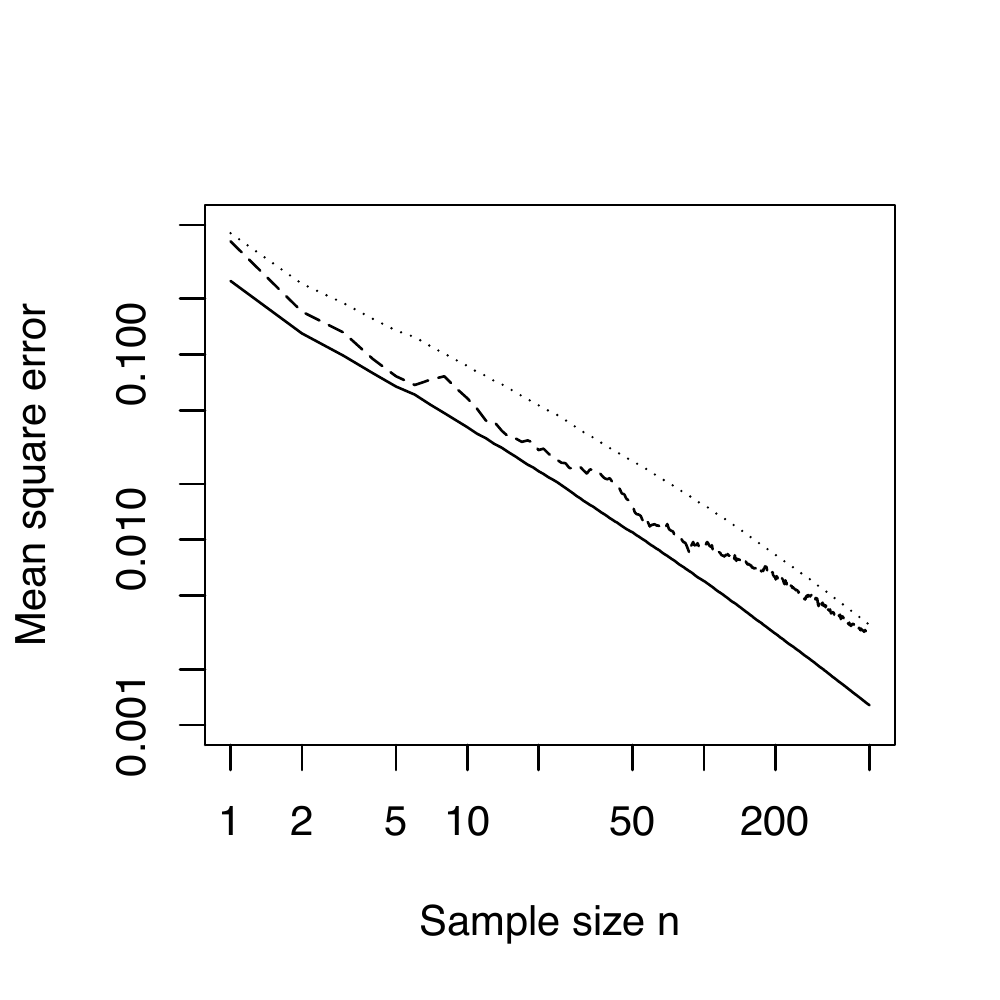}
\subcaption{Gender}
\end{subfigure}
\begin{subfigure}[b]{0.45\textwidth}
\centering
\includegraphics[width=\textwidth]{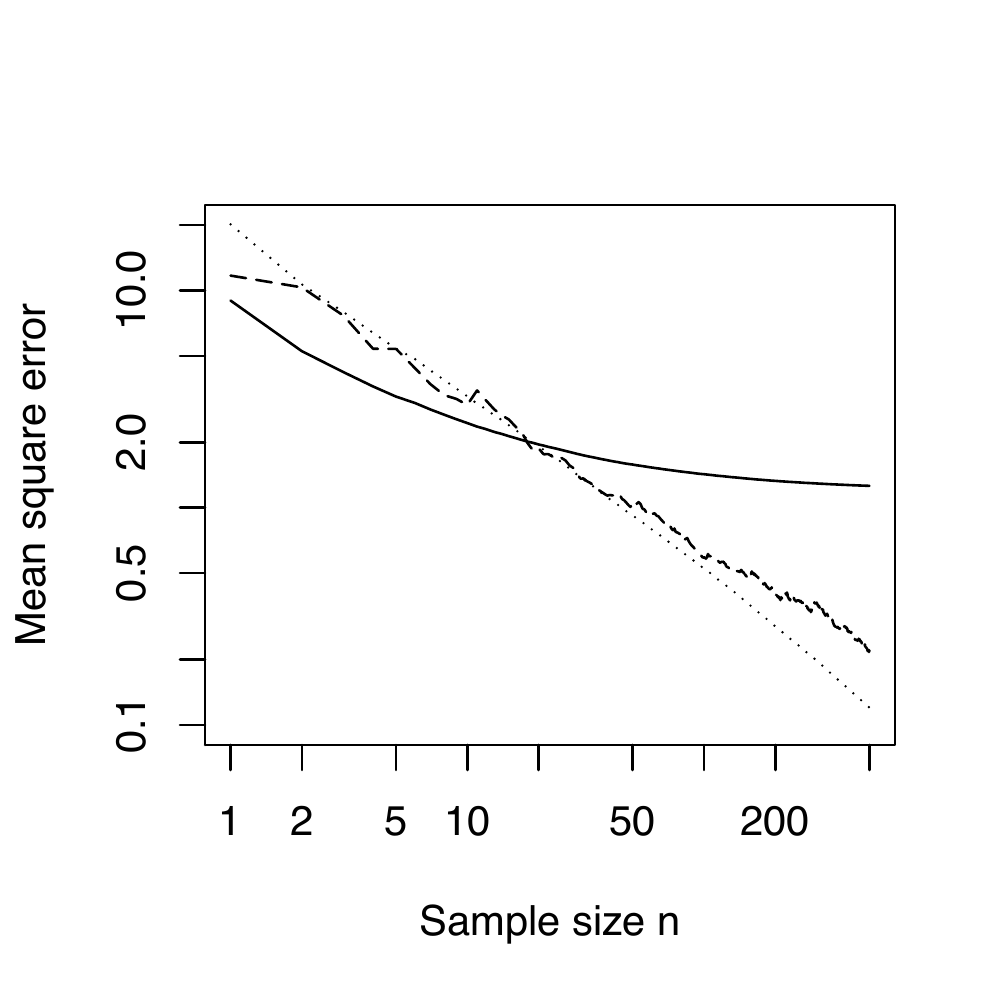}
\subcaption{Total nominations}
\end{subfigure}
\caption{The mean square error of $\hat{\mu}$ (solid), $\hat{\mu}_{IPW}$ (dotted), and $\hat{\mu}_{BA}$ (longdash) for gender and total nominations.}
\label{fig:mse}
\end{figure}


\section{Discussion}


A recent review of the RDS literature counted over 460 studies which used RDS \cite{white2015strengthening}.  Many of these studies seek to estimate the prevalence of HIV or other infectious diseases; for these studies, a point estimate of the prevalence is insufficient.  These studies have used confidence intervals constructed from bootstrap procedures and from estimates of the standard error \cite{RDSpackage}.  These standard error intervals implicitly rely on a central limit theorem and this paper provides a partial justification for such techniques, so long as $m \le 1/\lambda_2^2$.  This paper makes a first step at studying the distributional properties of two simple estimators in this regime.

Figure \ref{fig:simulation} suggests that if $m$ is larger than $1/\lambda_2^2$, then the simple estimators ($\hat \mu$ and $\hat \mu_{VH}$) are no longer normally distributed. Interestingly, under the simulation setting where the estimators are no longer normally distributed, the Q-Q plots are flatter than the line $x=y$.  This indicates that a confidence interval constructed from the standard errors would be conservative; a nominally  90\% confidence interval would cover $\mu$ more than 90\% of the time.  As such, a properly constructed interval should be narrower than the interval constructed from the standard error. If one pursues this path, then care must be taken in estimating the standard errors.  For example, a bootstrap procedure proposed in \cite{salg} has become very popular.  However, for reasons beyond the inequality in Proposition \ref{prop:auto}, this bootstrap procedure drastically underestimates the actual standard errors \cite{goel2010assessing, treevar}.

There are many reasons to suspect the construction of the sampling weights in RDS studies.  At the most basic level, the justification for 
\[\{\mbox{selection probability for node $i$}\} \propto deg(i)
\] 
comes from a Markov model which has several problematic pieces (e.g. replacement sampling, uniform selection of friends, referral process has reached equilibrium, and all network relationships are reciprocated). While these assumptions are all troublesome, they are merely sufficient conditions.  It is conceivable that $deg(i)$ is still an adequate approximation of the selection probabilities (up to scaling) even when the assumptions do not hold.  Perhaps the most difficult problem is that $deg(X_i)$ is estimated via self-report.  Taken together, there are many reasons to doubt the sampling weights.   In a related context, the third section of the paper discusses the bias-variance tradeoff between $\hat \mu$ and $\hat \mu_{IPW}$.  
The results in Proposition \ref{variancedifference} and Theorem \ref{variancedifferenceTree} presume that the sampling weights are known exactly.  However, given that the presumed model has several deficiencies, the stationary distribution of the (presumed) Markov process does not necessarily reveal the actual sampling probabilities.  As such,  $\hat \mu_{IPW}$ (and by extension $\hat \mu_{VH}$) are constructed with incorrect and noisy measurements of the sampling probabilities. This will likely make the estimator biased and more variable.  Because of this, in practice we should be less inclined towards the weighted estimator  (i.e. $\hat \mu_{VH}$) than the proposed estimator $\hat \mu_{BA}$ suggests.  Section \ref{sec:example} and the data analysis with the AddHealth network suggest that the sample average is perhaps less biased than was previously considered.  While there are certainly situations where bias corrected estimators should be used, it also seems sensible to first estimate the bias; when the bias is large, this is a relatively easy task.

The theorems in this paper do not apply to general trees, only to m-trees.  
If $\mathbb{T}$ is a Galton-Watson tree with $E(\eta(\sigma))<\lambda_2^{-2}$, 
then the simulations support the following conjecture:
$$\frac{1}{\sqrt{n}}\sum_{\sigma \in \mathbb{T}}(y(X_{\sigma})-E_{\pi}(y))\rightarrow N(0,\sigma^2),$$
where $\sigma^2$ could be computed from the results in \cite{treevar}. 
To prove this result requires a more careful study of the structure of $\{X_{\sigma}\}_{\sigma \in \mathbb{T}}$. We leave this problem to future investigation.

\section*{Acknowledgement}

This material is based upon work supported in part by the U. S. Army Research Office under grant number W911NF1510423 and the National Science Foundation under grant number DMS-1309998. 
We thank Zoe Russek for helpful comments. 
\appendix

\section{Proof of Theorem \ref{thm1}}\label{thm1proof}

In the appendix we give a proof of the theorems and propositions in the paper. First we give an outline of the proof of our main theorem. Consider the martingale $$\sum_{h}E(Y_h|\mathscr{F}_{h-1})-Y_h,$$ where $\{\mathscr{F}_h\}$ is a filtration to be defined later. Using the Markov property and the estimation of $var(Y_h)$, we show that the martingale difference sequence satisfies the condition of the martingale central limit theorem. 
In this section, $P$ will be a reversible transition matrix with eigenvalues  $1=\lambda_1 \geq |\lambda_2| \geq... \geq |\lambda_N|$ and corresponding eigenfunctions $f_1,..., f_N$ satisfying $ \sum_{k} f_i(k)f_j(k)\pi_k=\delta_{ij}$ for any $i,j$. We refer to \cite{levin2009markov} for the existence of such eigen-decomposition. Unless stated otherwise, expectations are calculated with respect to the tree indexed random walk on the graph. 

We begin with some lemma.

\begin{lemma} (Lemma 12.2 in \cite{levin2009markov}) \label{lem:spec}
Let $P$ be a reversible Markov transition matrix on the nodes in $G$ with respect to the stationary distribution $\pi$.  The eigenvectors of $P$, denoted as $f_1, \dots, f_{N}$, are real valued functions of the nodes $i \in G$ and orthonormal with respect to the inner product 
\begin{equation} \label{def:inner}
\langle f_a, f_b \rangle_\pi = \sum_{i \in G} f_a(i) f_b(i) \pi_i.
\end{equation}
If $\lambda$ is an eigenvalue of $P$, then $|\lambda|\le 1$.  The eigenfunction $f_1$ corresponding to the eigenvalue $1$ is taken to be taken to be the constant vector $\textbf{1}$.  If $X(0), \dots, X(t)$ represent $t$ steps of a Markov chain with transition matrix $P$, then  the probability of a transition from $i\in G$ to $j \in G$ in $t$ steps can be written as
\begin{equation}\label{eq:tsteps}
P(X(t) = j|X(0) = i) = P_{ij}^t =  \pi_j + \pi_j \sum_{\ell =2}^{N} \lambda_\ell^t f_\ell(i) f_\ell(j).
\end{equation}
\end{lemma}

\begin{lemma}
\label{cov}
For any nodes $\sigma,\tau$ in $\mathbb{T}$,
$$
cov(y(X_{\sigma}),y(X_{\tau}))=\sum_{l=2}^{N} \lambda_\ell^{d(\sigma, \tau)} <y,f_\ell>_{\pi}^2,
$$
where $<y,f_\ell>_{\pi}=\sum_{i=1}^{N} y(i)f_\ell(i)\pi_i$.
\end{lemma}

\begin{proof}
From the reversibility of the Markov chain and Lemma \ref{lem:spec}, we have
$$P(X_{\sigma}=j|X_{\tau}=i)=P_{ij}^{d(\sigma, \tau)}=\pi_j+\pi_j\sum_{l=2}^{N} \lambda_\ell^{d(\sigma, \tau)}f_\ell(i)f_\ell(j).$$
Therefore,
\begin{align*}
cov(y(X_{\sigma}),y(X_{\tau}))&=\sum_{i,j}y(i)y(j)\pi_iP(X_{\sigma}=j|X_{\tau}=i)-(\sum_{i}\pi_iy_i)^2\\
&= \sum_{i,j}y(i)y(j)\pi_i\pi_j \sum_{l=2}^{N} \lambda_\ell^{d(\sigma, \tau)}f_\ell(i)f_\ell(j) \\
&=\sum_{l=2}^{N} \lambda_\ell^{d(\sigma, \tau)} <y,f_\ell>_{\pi}^2,
\end{align*}
and the lemma is proved.\par
The next lemma gives the expression of $var(Y_h)$.
\end{proof}

\begin{lemma}[Variance of $Y_h$]
\label{var}
Suppose that $|\lambda_2|>0$. Then as $h\rightarrow \infty$,
\begin{align*}
var(Y_h)=\left \{
 \begin{array}{lcc}
{O(1)} &\text{if} &m<\lambda_2^{-2} \\
{O(h)} &\text{if} &m=\lambda_2^{-2} \\
{O((m\lambda_2^2)^h)} &\text{if} &m>\lambda_2^{-2}
\end{array}
\right. .
\end{align*}
\end{lemma}

\begin{proof}
For $k=0,1,...,h$, denote by $s_{hk}$ the number of ordered pairs $(\sigma, \tau)$ such that $|\sigma|=|\tau|=h$ and $d(\sigma,\tau)=2k$. Then $s_{h0}=m^h$, and $$s_{hk}=m^{h-k}m(m-1)(m^{k-1})^2=m^{h+k}-m^{h+k-1}$$for $k\geq 1$. By Lemma \ref{cov}, 
\begin{align*}
var(Y_h)=\frac{1}{m^h}\sum_{k=0}^{h} (s_{hk}\sum_{l=2}^{N} \lambda_\ell^{2k} <y,f_\ell>_{\pi}^2)&=\sum_{l=2}^{N} (<y,f_\ell>_{\pi}^2(1+\sum_{k=0}^{h}m^{k-1}(m-1)\lambda_\ell^{2k}))\\&\leq \sum_{l=2}^{N}  (<y,f_\ell>_{\pi}^2 \sum_{k=0}^{h}(m\lambda_2^2)^{k}).
\end{align*}
Thus
\begin{align*}
var(Y_h)=\left \{
 \begin{array}{lcc}
{O(1)} &\text{if} &m<\lambda_2^{-2} \\
{O(h)} &\text{if} &m=\lambda_2^{-2} \\
{O((m\lambda_2^2)^h)} &\text{if} &m>\lambda_2^{-2}
\end{array}
\right. .
\end{align*}

\end{proof}

\begin{corollary} 
\label{converge}
For any function $y$ on the state space, $\frac{1}{\sqrt{m^h}}Y_h \rightarrow 0$ in probability.\par
\end{corollary}

\begin{proof} It follows from Lemma \ref{var} that $$var(\frac{1}{\sqrt{m^h}}Y_h)=O(\lambda_2^{2h})\rightarrow 0.$$
\end{proof}

The next lemma is a convergence argument which we will use in the proof of Theorem \ref{thm1}.\par
\begin{lemma}[Slutsky's lemma]
\label{slutsky}
If $X_h\rightarrow X$ in distribution and $Y_h \rightarrow 0$ in probability, then
$X_h+Y_h \rightarrow X$ in distribution.
\end{lemma}

The following theorem from \cite{durrett} is essential to the proof of our main theorem.\par
\begin{theorem}[Martingale central limit theorem]
\label{MCLT}
Suppose that $\{Z_h\}_{h\geq 1}$ is adapted to the filtration $\{\mathscr{F}_h\}_{h\geq 1}$ and that $E(Z_{h+1}|\mathscr{F}_h)=0$ for all $h\geq 1$. Let $S_h=\sum_{i=1}^{h}Z_i$ and $V_h=\sum_{i=1}^{h}E(Z_i^2|\mathscr{F}_{i-1})$. If\par
(1) $V_h/h \rightarrow \sigma^2>0$ in probability and\par
(2) $h^{-1}\sum_{i=1}^{h}E(Z_i^2 1_{\{|Z_i|>\epsilon\sqrt{h}\}})\rightarrow 0$ for every $\epsilon>0$,\par
then $S_h/\sqrt{h}\rightarrow N(0,\sigma^2)$ in distribution.\par
\end{theorem}

Now we are ready to prove our main theorems.\par
\begin{proof}[\textbf{Proof of Theorem \ref{thm1}}]
Define $Y_h$ in the same way as Theorem \ref{thm1}. Without loss of generality, suppose that $E_{\pi}(y)=0$. Since $m<\lambda_2^{-2}$, $\sqrt{m}P-I$ is invertible. Let $y'=(\sqrt{m}P-I)^{-1}y$. Then $y'$ is also a function on the state space. We will first argue on the new node feature $y'$ and then convert back to $y$. Define $$Y_h'=\frac{1}{\sqrt{m^h}}\sum_{\tau \in \mathbb{T}: |\tau|=h} y'(X_{\tau}).$$ Let $z_{hk}=\sum_{\sigma:|\sigma|=h}1_{\{X_{\sigma}=k\}}$ for $h\geq 0$ and $k=1,..., N$. Define $z_h=\{z_{h1},...,z_{hN}\}$, and $$\mathscr{F}_h=\sigma(X_{\tau}: |\tau|\leq h)$$ for $h\geq 1$. It is obvious that  $\{Y_h\}_{h\geq 1}$ is adapted to the filtration $\{\mathscr{F}_h\}_{h\geq 1}.$ Let
$$Z_h=E(Y_h'|\mathscr{F}_{h-1})-Y_h'.$$Then $\{Z_h,\mathscr{F}_h\}_{h\geq 1}$ is a martingale difference sequence. We will verify that $\{Z_h,\mathscr{F}_h\}_{h\geq 1}$ satisfies (1) and (2) in Theorem \ref{MCLT}.

We have
\begin{align*}
Z_h&=mz_{h-1}Py'/\sqrt{m^h}-Y_h'=\frac{z_{h-1}^T(\sqrt{m}P)y'}{\sqrt{m^{h-1}}}-\frac{z_h^Ty'}{\sqrt{m^h}}.
\end{align*}

For any $\sigma \in \mathbb{T}$, denote by $p(\sigma)$ the parent node of $\sigma$. $Z_h$ can also be expressed as
\begin{align*}
Z_h=\sum_{\sigma:|\sigma|=h}E(\frac{y'(X_{\sigma})}{\sqrt{m^h}}|\mathscr{F}_{h-1})-\frac{y'(X_{\sigma})}{\sqrt{m^h}}=\sum_{\sigma:|\sigma|=h} \frac{1_{p(\sigma)}Py'-y'(X_{\sigma})}{\sqrt{m^h}}=\sum_{\sigma:|\sigma|=h}W_{\sigma},
\end{align*}
where $1_{p(\sigma)}$ is the $1\times N$ vector with $1_{p(\sigma),i}=1$ if $X_{p(\sigma)}=i$ and 0 otherwise.\par
We have 
$$E(W_{\sigma}|\mathscr{F}_{h-1})=0$$
and
$$E(W_{\sigma}^2|\mathscr{F}_{h-1})=\frac{Var_{p_i}(y')}{m^h}$$
for $i=X_{p(\sigma)}$, where $p_i$ is the $i$th row of the transition matrix $P$ and $Var_{p_i}(y')=\sum_{j}p_{ij}(y(j)-\sum_{j}p_{ij}y(j))^2$. From the definition of tree indexed Markov process, if $|\sigma|=|\tau|=h$, then $W_{\sigma},W_{\tau}$ are independent given $\{X_{\sigma}:|\sigma|=h-1\}$. Using $E(W_{\sigma}|\mathscr{F}_{h-1})=0$, we have
\begin{align*}
E(Z_h^2|\mathscr{F}_{h-1})=\sum_{\sigma:|\sigma|=h}E(W_{\sigma}^2|\mathscr{F}_{h-1})=\sum_{i=1}^{N}\frac{z_{hi}}{m^h}Var_{p_i}(y').\\
\end{align*}
From Corollary \ref{converge}, $var(\frac{z_{hi}}{m^h})=O(\lambda_2^{2h})\rightarrow 0$ for every $i$. Thus $var(E(Z_h^2|\mathscr{F}_{h-1}))=O(\lambda_2^{2h})\rightarrow 0 (h\to \infty)$ and $\sum_{i=1}^{h}var(E(Z_i^2|\mathscr{F}_{i-1}))$ converges. It follows from the definition of $V_h$ and the Cauchy-Schwarz inequality that
$$\lim_{h \to \infty}var(V_h/h)=\lim_{h \to \infty}var(\frac{1}{h}\sum_{i=1}^{h}E(Z_i^2|\mathscr{F}_{i-1}))\leq\lim_{h \to \infty}\frac{1}{h}\sum_{i=1}^{h}var(E(Z_i^2|\mathscr{F}_{i-1}))=0.$$ Therefore$$V_h/h\rightarrow \sigma^2$$in probability, where
\begin{align*}
\sigma^2=E(E(Z_h^2|\mathscr{F}_{h-1}))=\sum_{i=1}^{N}\pi_i Var_{p_i}(y')&=\sum_{i=1}^{N}\pi_i(\sum_{j=1}^{N} p_{ij} y'(j)^2-(\sum_{j=1}^{N} p_{ij}y'(j))^2)\\&=var_{\pi}(y')-var_{\pi}(Py'),
\end{align*}
and condition (1) in Theorem \ref{MCLT} is satisfied.

Similarly, we have
\begin{align*}
E(Z_h^4|\mathscr{F}_{h-1})
&=\sum_{\sigma:|\sigma|=h}E(W_{\sigma}^4|\mathscr{F}_{h-1})+\sum_{\sigma,\tau:|\sigma|=|\tau|=h}E(W_{\sigma}^2W_{\tau}^2|\mathscr{F}_{h-1})\\
&\leq \frac{C_0}{m^h}+C_1\frac{m^h(m^h-1)}{m^{2h}}\leq C,\\
\end{align*}
where $C_0,C_1,C$ are constants. Thus $E(Z_h^4)\leq C$ for any $h$, and
$$h^{-1}\sum_{i=1}^{h}E(X_i^2 1_{\{|X_i|>\epsilon\sqrt{h}\}})\leq h^{-1}\sum_{i=1}^{h}E(X_i^2 \frac{X_i^2}{\epsilon^2 h})=\frac{1}{\epsilon^2 h^2}\sum_{i=1}^{h}E(X_i^4)
\leq\frac{C}{h}\rightarrow 0.$$
Condition (2) is also satisfied. From Theorem \ref{MCLT}, we have$$\frac{1}{\sqrt{h}}\sum_{i=1}^{h}Z_i=\frac{1}{\sqrt{h}}\sum_{i=1}^{h}(\frac{z_{i-1}^T(\sqrt{m}P)y'}{\sqrt{m^{i-1}}}-\frac{z_i^Ty'}{\sqrt{m^i}}) \rightarrow N(0,\sigma^2)$$ in distribution. If $m<\lambda_2^{-2}$, then from Lemma \ref{var}, $\frac{z_{h}^T(\sqrt{m}P)y'}{\sqrt{m^{h-1}}} \rightarrow 0$ in probability. From Lemma \ref{slutsky} and the definition of $y'$,
\begin{align*}
\frac{1}{\sqrt{h}}\sum_{i=1}^{h}\frac{z_{i}^T(\sqrt{m}P-I)y'}{\sqrt{m^i}}=\frac{1}{\sqrt{h}}\sum_{i=1}^{h}Y_i\rightarrow N(0,\sigma^2)
\end{align*}
in distribution, where $\sigma^2=var_{\pi}(y')-var_{\pi}(Py')=var_{\pi}((\sqrt{m}P-I)^{-1}y)-var_{\pi}(P(\sqrt{m}P-I)^{-1}y).$ The proof is now complete.\par
\end{proof}


\section{Proof of Theorem \ref{thm2} and Corollary \ref{cor:vh}}\label{thm2proof}

We provide a proof of the central limit theorem using the moment method. It involves a careful study of all the moments of $\hat\mu_h$. The following proposition is essential to our proof.

\begin{proposition}
(Moment continuity theorem) Let ${X_h}$ be a sequence of uniformly subgaussian real random variables, and let $X$ be another subgaussian random variable. Then the following statements are equivalent:

(1)$EX_h^k\rightarrow EX^k$ for all $k$

(2)$X_h \rightarrow X$ in distribution

\end{proposition}

Following this proposition, we can break down our proof to two parts. We will first prove that all the moments of $\hat{\mu}_h$ converge to the moments of some normal distribution. Then we will verify that $\hat{\mu}_h$ is a uniformly subgaussian sequence.

\subsection{Proof of moments convergence}
Let $X_r$ be the root of the 2-tree. Define
$$\gamma_{k,h}(i)=E[\hat{\mu}_h^k|X_{r}=i],$$ and
$$\gamma_{k,h}=E[\hat{\mu}_h^k].$$
Let $\rho=\sqrt{2}|\lambda_2|<1$. We will prove that there exist $\gamma_{k}$, $k\geq 1$ such that $|\gamma_{k,h}(i)-\gamma_{k}|=O(\rho^{h})$ for all  $k,i,h$ and that $\gamma_{k}=E(\xi^{k})$ for $\xi\sim N(0,\gamma_2)$.
Our key observation is that the left and right subtree can be seen as i.i.d copies of the whole tree given the left and right child of the seed, which makes it possible to build a relationship between $\gamma_{k,h}(i)$ and $\gamma_{k,h-1}(i)$. Only condition (c3) is needed throughout the proof.

We need the following Lemma.
\begin{lemma}
\label{sequenceconverge}
Let $\{a_h\}$ be a sequence satisfying
$$a_h=c_h(ca_{h-1}+C+d_h),$$
where $|1-c_h|=O(\rho^h)$, $|d_h|=O(\rho^h)$, $C$ is a constant and $c<\rho<1$. Then $|a_h-C/(1-c)|=O(\rho^h)$.
\end{lemma}

\begin{proof}
Without loss of generality, suppose that $c_h\neq 0$ and $C=0$. Since $|1-c_h|=O(\rho^h)$ and $|d_h|=O(\rho^h)$, there exists $M$ such that $\prod_{k=1}^{h}|c_k|\leq M$ and $|d_h|\leq M\rho^h$ for all $h$. Therefore,
$$|a_h|=\sum_{k=0}^{h}c^{h-k}d_k\prod_{i=h-k}^{h}|c_i|\leq \frac{M^2}{\rho-c}\rho^h,$$
and the lemma is proved.
\end{proof}

We use an induction on $k$. First, we will prove that $\gamma_1=0$. In fact, from Lemma \ref{lem:spec}, 
$$E(y(X_{\sigma})|X_{r}=i)=\sum_{j=1}^{N}y(j)P_{ij}^{|\sigma|}=O(|\lambda_2|^{|\sigma|}).$$
Therefore $$E(\hat{\mu}_h|X_{r}=i)=\frac{1}{\sqrt{2}^h}\sum_{k=1}^{h}2^{k}O(|\lambda_2|^k)=O(\rho^{h})$$
for all i, and $|\gamma_{1,h}(i)-\gamma_1|=O(\rho^{h})$ for $\gamma_1=0$.

Now we move from $k-1$ to $k$. Without loss of generality, suppose that $\gamma_{2,h}(i)>1$ for all $h,i$ (or we can multiply $y$ with a large constant). It follows that $\gamma_{2k,h}(i)\geq(\gamma_{2,h}(i))^k>1$ for all $k$.
We can decompose $\gamma_{k,h}(i)$ into
\begin{equation}
\begin{aligned}
\label{decompose}
\gamma_{k,h}(i)&=E[\hat{\mu}_h^k|X_{r}=i]-E[(\frac{1}{\sqrt{2}^h}\sum_{\sigma\in \mathbb{T}, 0<|\sigma|\leq h}X_{\sigma})^k|X_{r}=i]\\
&+E[(\frac{1}{\sqrt{2}^h}\sum_{\sigma\in \mathbb{T}, 0<|\sigma|\leq h}X_{\sigma})^k|X_{r}=i]\\
&:=I_1+I_2
\end{aligned}
\end{equation}
If $k$ is even, then from Holder's inequality we know that 
\begin{equation}
\begin{aligned}
\label{Holdereven}
|I_1|=&|E[\hat{\mu}_h^k|X_{r}=i]-E[(\hat{\mu}_h-\frac{i(s)}{\sqrt{2^{h}}})^k|X_{r}=i]|\\
=&\sum_{m=1}^{k}\binom{k}{m}(-i\sqrt{2}^{-h})^{m}E[\hat{\mu}_h^{k-m}|X_{r}=i]\\
\leq& \sum_{m=1}^{k}\binom{k}{m}(|i|\sqrt{2}^{-h})^{m}E[\hat{\mu}_h^k|X_{r}=i]\\
\leq& [(1+M\sqrt{2}^{-h})^{k}-1]E[\hat{\mu}_h^k|X_{r}=i].\\
\end{aligned}
\end{equation}
Likewise, If $k$ is odd, we have
\begin{equation}
\begin{aligned}
\label{Holderodd}
|I_1|=&E[\hat{\mu}_h^k|X_{r}=i]-E[(\hat{\mu}_h-\frac{i}{\sqrt{2^{h}}})^k|X_{r}=i]\\
\leq& [(1+M\sqrt{2}^{-h})^{k}-1]E[\hat{\mu}_h^{k-1}|X_{r}=i]\\
\end{aligned}
\end{equation}
Since $k$ is fixed, $E[\hat{\mu}_h^{k-1}|X_{r}=i]$ is bounded from our assumption on $\gamma_{k-1,h}(i)$, and 
$(1+M\sqrt{2}^{-h})^{k}-1=O(\sqrt{2}^{-h})=O(\rho^h).$ Hence,
$$
|I_1|\leq[(1+M\sqrt{2}^{-h})^{k}-1]E[\hat{\mu}_h^{k-1}|X_{r}=i]=O(\rho^h).
$$

Let $X_{lc}$ and $X_{rc}$ be the left and right child of the root and $\mathbb{T}_l$ and $\mathbb{T}_r$ the left and right subtree, we have
\begin{equation}
\begin{aligned}
\label{subtree}
&E[(\hat{\mu}_h-\frac{i}{\sqrt{2^{h}}})^k|X_{r}=i]\\
=&E((\frac{1}{\sqrt{2}}(\frac{1}{\sqrt{2}^{h-1}}\sum_{\sigma\in \mathbb{T}_r, |\sigma|\leq h}X_{\sigma}+\frac{1}{\sqrt{2}^{h-1}}\sum_{\sigma\in \mathbb{T}_l, |\sigma|\leq h}X_{\sigma}))^k|X_{r}=i)\\
=&\sum_{u,v}p_{iu}p_{iv}E((\frac{1}{\sqrt{2}}(\frac{1}{\sqrt{2}^{h-1}}\sum_{\sigma\in \mathbb{T}_r, |\sigma|\leq h}X_{\sigma}+\frac{1}{\sqrt{2}^{h-1}}\sum_{\sigma\in \mathbb{T}_l, |\sigma|\leq h}X_{\sigma}))^k|X_r=i, X_{lc}=u, X_{rc}=v)\\
=&\sum_{u,v}p_{iu}p_{iv}E((\frac{1}{\sqrt{2}}(\frac{1}{\sqrt{2}^{h-1}}\sum_{\sigma\in \mathbb{T}_r, |\sigma|\leq h}X_{\sigma}+\frac{1}{\sqrt{2}^{h-1}}\sum_{\sigma\in \mathbb{T}_l, |\sigma|\leq h}X_{\sigma}))^k|X_{lc}=u, X_{rc}=v)\\
=&\sum_{u,v}p_{iu}p_{iv}\frac{1}{\sqrt{2}^{k}}(\gamma_{k,h-1}(u)+\gamma_{k,h-1}(v))+S_1\\
=&\frac{1}{\sqrt{2}^{k-2}}\sum_{u}p_{iu}\gamma_{k,h-1}(u)+S_1,\\
\end{aligned}
\end{equation}
where
$$S_1=\sum_{m=1}^{k-1}\frac{1}{\sqrt{2}^{k}}\binom{k}{m}\sum_{u,v}p_{iu}p_{iv}\gamma_{m,h-1}(u)\gamma_{k-m,h-1}(v).$$

If $k=2$, Equation \ref{Holdereven} and \ref{subtree} reduce to
$$
\gamma_{k,h}(i)=\sum_{u}p_{iu}\gamma_{k,h-1}(u)+\delta_{h}(i),
$$
where
$$\delta_{h}(i)=\frac{y(i)^2}{2^h}-\frac{2y(i)}{\sqrt{2}^h}\gamma_{1,h}(i)+(\sum_{u}p_{iu}\gamma_{1,h}(u))^2=O(\rho^h).
$$
Write $\nu_h=\{\gamma_{h}^2\}'$ and $\delta_{h}=\{\delta_{h}\}'$. For $n\geq 2$,
$$\nu_h=P\nu_{h-1}+\delta_{h}.$$
Thus by setting $\delta_1=0$ we have $\nu_h=P^h\nu_1+\sum_{k=1}^{h}P^k\delta_{h-k}$,and it is not hard to verify that all the components of $\nu_h$ (i.e., every $\gamma_{k,h}(i)$) converge to $\gamma_2=\pi^t\nu_1+\sum_{h=1}^{\infty}\pi^t\delta_h$ with rate $\rho^h$.

Now suppose that $k>2$. There are a fixed number of terms in $S_1$. Since $|\gamma_{l,h}(i)-\gamma_{l}|=O(\rho^{h})$ for all $i\in S$ and $l<k-1$, we have
$$|S_1-\sum_{m=1}^{k-1}\frac{1}{\sqrt{2}^{k}}\binom{k}{m}\gamma_{m}\gamma_{k-m}|=O(\rho^{h}).$$Thus,
\begin{equation}
\label{term2approx}
I_2=\frac{1}{\sqrt{2}^{k-2}}\sum_{u}p_{iu}\gamma_{k,h-1}(u)+\sum_{m=1}^{k-1}\frac{1}{\sqrt{2}^{k}}\binom{k}{m}\gamma_{m}\gamma_{k-m}+O(\rho^{h}).
\end{equation}

Combining Equation \ref{Holdereven}, \ref{Holderodd} and \ref{term2approx} we arrive at the final equation for $\gamma_{k,h}(i)$:
\begin{equation}
\gamma_{k,h}(i)=c_{k,h}I_2=c_{k,h}(\frac{1}{\sqrt{2}^{k-2}}\sum_{u}p_{iu}\gamma_{k,h-1}(u)+\sum_{m=1}^{k-1}\frac{1}{\sqrt{2}^{k}}\binom{k}{m}\gamma_{m}\gamma_{k-m}+O(\rho^{h})),
\end{equation}
where
$c_{k,h}=1$ if $k$ is odd and $c_{k,h}=1+O(\rho^h)\in[2-(1+M\sqrt{2}^{-h})^{k},(1+M\sqrt{2}^{-h})^{k}]$ if $k$ is even.
Since $\frac{1}{\sqrt{2}^{k-2}}<\rho$ and $\prod_{i=1}^{h}c_{k,h}$ converges, we conclude from Lemma  \ref{sequenceconverge} that
$$|\gamma_{k,h}(i)-\gamma_{k}|=O(\rho^{h}),$$ where
$\gamma_{k}=\sum_{m=1}^{k-1}\frac{1}{\sqrt{2}^{k}}\binom{k}{m}\gamma_{m}\gamma_{k-m}/(1-\frac{1}{\sqrt{2}^{k-2}})$.

We have proved that $|\gamma_{k,h}(i)-\gamma_{k}|=O(\rho^{h})$ for all  $k,i,h$. Let $h$ tend to infinity in Equation \ref{term2approx}, we have 
\begin{equation}
\label{relation}
\gamma_{k}=\frac{1}{\sqrt{2}^{k-2}}\gamma_{k}+\sum_{m=1}^{k-1}\frac{1}{\sqrt{2}^{k}}\binom{k}{m}\gamma_{m}\gamma_{k-m}.
\end{equation}

Now suppose that $\xi_i$ is a sequence of i.i.d $N(0, \gamma_2)$ variables. Let
$$\tilde\gamma_{k}=\lim_{h \to \infty}E((\frac{\xi_1+\dots+\xi_h}{\sqrt{h}})^k),$$
then $\{\tilde\gamma_{k}\}, k\in \mathbb{N}$ also follows Equation \ref{relation}. Since $\gamma_1=\tilde\gamma_1=0$ and $\gamma_2=\tilde\gamma_2$ we have $\gamma_{k}=\tilde \gamma_{k}$ for every $k$, and the argument is proved.

\subsection{Proof of uniform subgaussianity}
To prove that $\hat{\mu}_h$ are uniformly subgaussian for all $h$, we need to show that there exists some $\theta$ such that
$$\gamma_{2\ell,h}(i)\leq \theta^{2\ell}\gamma_{2\ell}$$
for all $\ell$. 

Let $c_1$ be a large constant to be defined later and $c_{h+1}=(1+M(2\lambda_2')^{-h})(1+(\sqrt{2}\lambda_2')^{h})c_h$, where $\lambda_2'=max\{|\lambda_2|,2/3\}$ and $M=||y||_{\infty}$. Let $s_{\ell,h}=||\gamma_{\ell,h}(i)||_{\infty}$. Since $0<\sqrt{2}\lambda_2'<1<2\lambda_2'$, $c_{h+1}>c_h$ and $\theta=\lim_{h \to \infty}c_h$ exists. Thus it suffices to prove that
\begin{equation}
\label{supeven}
s_{2\ell,h}\leq c_h ^{2\ell}\gamma_{2\ell}
\end{equation}
and
\begin{equation}
\label{supodd}
s_{2\ell-1,h}\leq c_h ^{2\ell-1}(\sqrt{2}\lambda_2')^h \gamma_{2\ell}
\end{equation}
for all $\ell$ and $h$.

Again we use an induction on $\ell$. Since $\gamma_{1,h}(i)=O(|\lambda_2|^h)$, we can choose $c_1$ large enough such that the inequalities in Equation \ref{supeven} and \ref{supodd} hold for all $(h,\ell)$ with $h=1$ or $\ell=1$. Suppose that \ref{supeven} and \ref{supodd} are verified for all $\ell\leq k$. We will prove that they are also true for $\ell=k+1$. 

From condition (c1) and (c2), we know that $$||\sum_{u}p_{iu}\gamma_{2k+1,h}(u)||_{\infty}\leq\lambda_2 ||\gamma_{2k+1,h}(i)||_{\infty}=\lambda_2 s_{2k+1,h}.$$

From our assumption of induction we have
\begin{equation}
\begin{aligned}
s_{2k+1,h}&\leq[(1+M\sqrt{2}^{-h})^{2k+1}-1]c_{h-1} ^{2k}\gamma_{2k}+c_{h-1}^{2k+1}\frac{(|\sqrt{2}\lambda_2'|)^{h}}{\sqrt{2}^{2k+2}}(\binom{2k+1}{0}\gamma_{2k+2}\\
&+\sum_{m=1}^{k}(\binom{2k+1}{2m-1}+\binom{2k+1}{2m})\gamma_{2k+2-2m}\gamma_{2m}+\binom{2k+1}{2k+1}\gamma_{2k+2})\\
&= [(1+M\sqrt{2}^{-h})^{2k+1}-1]c_{h-1} ^{2k}\gamma_{2k}+c_{h-1} ^{2k+1}\\
&\frac{(|\sqrt{2}\lambda_2'|)^{h}}{\sqrt{2}^{2k+2}}(\binom{2k+2}{0}\gamma_{2k+2}+\sum_{m=1}^{k}\binom{2k+2}{2m}\gamma_{2k+2-2m}\gamma_{2m}+\binom{2k+2}{2k+2}\gamma_{2k+2})\\
&=[(1+M\sqrt{2}^{-h})^{2k+1}-1]c_{h-1} ^{2k}\gamma_{2k}+c_{h-1} ^{2k+1}(|\sqrt{2}\lambda_2'|)^{h}\gamma_{2k+2}\\
&\leq([(1+M\sqrt{2}^{-h})^{2k+1}-1](|\sqrt{2}\lambda_2'|)^{-h}+1)c_{h-1} ^{2k+1}(|\sqrt{2}\lambda_2'|)^{h}\gamma_{2k+2}\\
&\leq(1+M(2\lambda_2')^{-h})^{2k+1}c_{h-1} ^{2k+1}(|\sqrt{2}\lambda_2'|)^{h}\gamma_{2k+2}\\
&\leq c_h^{2k+1}(|\sqrt{2}\lambda_2'|)^{h}\gamma_{2k+2},
\end{aligned}
\end{equation}
and Equation \ref{supodd} is true for $2k+1$.

Now we move from $2k+1$ to $2k+2$. Recall that
\begin{align*}
&\gamma_{2k+2,h}(i)=c_{2k+2,h}(\sum_{m=0}^{2k+2}\frac{1}{\sqrt{2}^{2k+2}}\binom{2k+2}{m}\sum_{u,v}p_{iu}p_{iv}\gamma_{m,h-1}(u)\gamma_{2k+2-m,h-1}(v))\\
\leq&(1+M2^{-h})^{2k+2}(\sum_{m=0}^{2k+2}\frac{1}{\sqrt{2}^{2k+2}}\binom{2k+2}{m}s_{m,h-1}s_{2k+2-m,h-1}).
\end{align*}

Thus,
$$
s_{2k+2,h}\leq (1+M2^{-h})^{2k+2}(\sum_{m=0}^{2k+2}\frac{1}{\sqrt{2}^{2k+2}}\binom{2k+2}{m}s_{m,h-1}s_{2k+2-m,h-1}).
$$
Let
$$I_1=\sum_{m=0}^{k+1}\frac{1}{\sqrt{2}^{2k+2}}\binom{2k+2}{2m}s_{2m,h-1}s_{2k+2-2m,h-1}$$
and
$$I_2=\sum_{m=0}^{k}\frac{1}{\sqrt{2}^{2k+2}}\binom{2k+2}{2m+1}s_{2m+1,h-1}s_{2k+1-2m,h-1}.$$
Then $s_{2k+2,h}\leq (1+M2^{-h})^{2k+2}(I_1+I_2)$. 

We have
\begin{equation}
\begin{aligned}
\label{subgaussterm1}
I_1 &\leq \sum_{m=0}^{k+1}\frac{1}{\sqrt{2}^{2k+2}}\binom{2k+2}{2m}c_{h-1}^{2m}\gamma_{2m}c_{h-1}^{2k+2-2m}\gamma_{2k+2-2m}\\
&=c_{h-1}^{2k+2}\sum_{m=0}^{k+1}\frac{1}{\sqrt{2}^{2k+2}}\binom{2k+2}{2m}\gamma_{2m}\gamma_{2k+2-2m}\\
&=c_{h-1}^{2k+2}\gamma_{2k+2},\\
\end{aligned}
\end{equation}
where the last equality follows from Equation \ref{relation}.
On the other hand
\begin{equation}
\begin{aligned}
I_2&\leq \sum_{m=0}^{k}\frac{1}{\sqrt{2}^{2k+2}}\binom{2k+2}{2m+1}c_{h-1}^{2m+1}(\sqrt{2}\lambda_2')^{h-1}\gamma_{2m+2}c_{h-1}^{2k+1-2m}(\sqrt{2}\lambda_2')^{h-1}\gamma_{2k+2-2m}\\
&=c_{h-1}^{2k+2}(\sqrt{2}\lambda_2')^{2(h-1)}\sum_{m=0}^{k}\frac{1}{\sqrt{2}^{2k+2}}\binom{2k+2}{2m+1}\gamma_{2m+2}\gamma_{2k+2-2m}.\\
\end{aligned}
\end{equation}
It can be directly verified that for all $m$,
$$
\binom{2k+2}{2m+1}\gamma_{2m+2}\gamma_{2k+2-2m}\leq2(k+1)\binom{k}{m}\gamma_{2k+2}.
$$
Thus,
\begin{equation}
\begin{aligned}
\label{subgaussterm2}
I_2&\leq c_{h-1}^{2k+2}(\sqrt{2}\lambda_2')^{h-1}\sum_{m=0}^{k}\frac{1}{\sqrt{2}^{2k+2}}2(k+1)\binom{k}{m}\gamma_{2k+2}\\
&=c_{h-1}^{2k+2}(\sqrt{2}\lambda_2')^{2(h-1)}(k+1)\gamma_{2k+2}\\
&\leq c_{h-1}^{2k+2}(\sqrt{2}\lambda_2')^{h}2k\gamma_{2k+2}.
\end{aligned}
\end{equation}
Combining Equation \ref{subgaussterm1} and \ref{subgaussterm2}, we have
\begin{equation}
\begin{aligned}
I_1+I_2&\leq c_{h-1}^{2k+2}\gamma_{2k+2}(1+2k*(\sqrt{2}\lambda_2')^{h})\\
&\leq c_{h-1}^{2k+2}\gamma_{2k+2}(1+(\sqrt{2}\lambda_2')^{h})^{2k+2}.\\
\end{aligned}
\end{equation}
Therefore,
\begin{equation}
\begin{aligned}
s_{h,2k+2}&\leq c_{2k+2,h}(I_1+I_2)\\
&\leq c_{h-1}^{2k+2}(1+M2^{-h})^{2k+2}(1+(\sqrt{2}\lambda_2')^{h})^{2k+2}\gamma_{2k+2}\\
&\leq c_h^{2k+2}\gamma_{2k+2},\\
\end{aligned}
\end{equation}
and the theorem is proved.

\subsection{Proof of Corollary \ref{cor:vh}}

\begin{proof}
By Theorem \ref{thm2} and Slutsky's lemma, it suffices to prove that $\hat{\bar{d}} \rightarrow \bar d$ in probability. Let $D=max_{1\leq i \leq N}deg(X_i)$. For any $\sigma \in \mathbb{T}$, $E_{\pi}\frac{1}{deg(X_{\sigma})}=\frac{1}{\bar d}$. Thus 
$E\frac{1}{\hat{\bar{d}}}=\frac{1}{\bar d}$, and it follows from Theorem \ref{thm2} that ${\hat {\bar {d}}}\rightarrow {\bar d}$ in probability. Since $\hat{\bar{d}}, \bar d \leq D$, we have $P(|\hat{\bar{d}}-\bar d|>\epsilon)\leq P(|\frac{1}{\hat{\bar{d}}}-\frac{1}{\bar d}|>\epsilon/D^2) \rightarrow 0$ for all $\epsilon>0$, and the corollary is proved.

\end{proof}

\section{Proof of Proposition \ref{variancedifference} and Theorem \ref{variancedifferenceTree}}\label{propproof}
\begin{proof} [\textbf{Proof of Proposition \ref{variancedifference}}]
Since $y(1),...,y(N)$ are uncorrelated, we have
$$
E(y(X_i)y(X_j))=\left \{
 \begin{array}{lcc}
{\mu_2\sum_{i=1}^{N}\pi_i^2} &\text{if} &i\neq j \\
{\mu_2} &\text{if} &i=j \\
\end{array}
\right. .
$$
Therefore,
\begin{align*}
var(\hat{\mu})=\frac{\mu_2}{n}-\frac{\mu_1^2}{n}+\frac{n-1}{n}(\mu_2\sum_{i=1}^{N}\pi_i^2+\mu_1^2(1-\sum_{i=1}^{N}\pi_i^2)-\mu_1^2).
\end{align*}
Similarly, we have
\begin{align*}
var(\hat{\mu}_{IPW})&=\frac{1}{n}var(\frac{y(X_1)}{N\pi_{X_1}})+\frac{n-1}{n}cov(\frac{y(X_1)}{N\pi_{X_1}}\frac{y(X_2)}{N\pi_{X_2}})\\&=\frac{\mu_2}{n}\sum_{i=1}^{N}\frac{1}{N^2\pi_i}-\frac{\mu_1^2}{n}+\frac{n-1}{n}(\frac{\mu_2}{N}+\mu_1^2(1-\frac{1}{N})-\mu_1^2).
\end{align*}Let
\begin{equation}
\begin{aligned}
\label{VD}
VD&=var(\hat{\mu}_{IPW})-var(\hat{\mu})\\
&=\frac{\mu_2}{n}(\sum_{i=1}^{N}\frac{1}{N^2\pi_i}-1)+\frac{n-1}{n}var(y)(\frac{1}{N}-\sum_{i=1}^{N}\pi_i^2).
\end{aligned}
\end{equation}
Since $max_{1\leq i\leq N}N\pi_i=C_1$, $N^2\pi_i\pi_j\leq {C_1^2}$ for all $i,j$. Therefore
\begin{equation}
\label{term1}
\begin{aligned}
\sum_{i=1}^{N}\frac{1}{N^2\pi_i}-1=\frac{1}{2}\sum_{i=1}^{N}\sum_{j=1}^{N}(\frac{\sqrt{\pi_j}}{N\sqrt{\pi_i}}-\frac{\sqrt{\pi_i}}{N\sqrt{\pi_j}})^2&=\frac{1}{2}\sum_{i=1}^{N}\sum_{j=1}^{N}\frac{(\pi_i-\pi_j)^2}{N^2\pi_i\pi_j}\\
&\geq\frac{1}{2C_1^2}\sum_{i=1}^{N}\sum_{j=1}^{N}{(\pi_i-\pi_j)^2}.
\end{aligned}
\end{equation}
We have
\begin{equation}
\begin{aligned}
\label{term2}
\frac{1}{2}\sum_{i=1}^{N}\sum_{j=1}^{N}{(\pi_i-\pi_j)^2}=\frac{1}{2}\sum_{i=1}^{N}\sum_{j=1}^{N}(\pi_i^2+\pi_j^2)-\sum_{i=1}^{N}\sum_{j=1}^{N}\pi_i\pi_j&=N\sum_{i=1}^{N}\pi_i^2-1\\
&=N^2var(\pi).
\end{aligned}
\end{equation}
From Eq.\ref{VD}, \ref{term1} and \ref{term2},
$$
VD\geq(\frac{\mu_2N^2}{C_1^2n}-Nvar(y))var(\pi),
$$
and the proposition is proved.
\end{proof}

\begin{proof} [\textbf{Proof of Theorem \ref{variancedifferenceTree}}]
For the $(\mathbb{T},P)$-walk on $G$, the variance of $\hat{\mu}_{IPW}$ and $\hat{\mu}$ is given by
\begin{align*}
var(\hat{\mu}_{IPW})=\frac{\mu_2}{n}\sum_{i=1}^{N}\frac{1}{N^2\pi_i}-\frac{\mu_1^2}{n}+\frac{1}{n^2}\sum_{\sigma\neq \tau}var(y)\sum_{i=1}^{N}\frac{p_{ii}^{d(\sigma, \tau)}}{N^2\pi_i}
\end{align*}
and
\begin{align*}
var(\hat{\mu})=\frac{\mu_2}{n}-\frac{\mu_1^2}{n}+\frac{1}{n^2}\sum_{\sigma\neq \tau}var(y)\sum_{i=1}^{N}\pi_i{p_{ii}^{d(\sigma, \tau)}}.
\end{align*}
Let
\begin{equation}
\begin{aligned}
\label{VDTree}
VD&=var(\hat{\mu}_{IPW})-var(\hat{\mu})\\
&=\frac{\mu_2}{n}(\sum_{i=1}^{N}\frac{1}{N^2\pi_i}-1)+\frac{1}{n^2}var(y)\sum_{\sigma\neq \tau}\sum_{i=1}^{N}(\frac{p_{ii}^{d(\sigma, \tau)}}{N^2\pi_i}-\pi_i{p_{ii}^{d(\sigma, \tau)}}).
\end{aligned}
\end{equation}
Since $C_1N\leq d_i\leq C_2N$, we have $\pi_i \leq \frac{C_2}{NC_1}$ for all $i$ and $p_{ij}^{(t)}\leq \frac{1}{C_1N}$ for all $i,j,t$. Thus,
$$
\frac{1}{n^2}\sum_{\sigma\neq \tau}\sum_{i=1}^{N}(\frac{p_{ii}^{d(\sigma, \tau)}}{N^2\pi_i}-\pi_i{p_{ii}^{d(\sigma, \tau)}})\geq \frac{1}{n^2}\sum_{\sigma\neq \tau}\sum_{i=1}^{N}\pi_i{p_{ii}^{d(\sigma, \tau)}} \geq \frac{C_2}{NC_1^2}.
$$
From equation \ref{term1} and \ref{term2}, we know that
$$
\sum_{i=1}^{N}\frac{1}{N^2\pi_i}-1\geq \frac{N^2C_1^2}{C_2^2}var(\pi).
$$
Therefore,
$$
VD\geq \mu_2(\frac{N^2C_1^2}{nC_2^2}var(\pi)-\frac{C_2}{NC_1^2}).
$$
We conclude the proof by taking $C=\frac{C_2^3}{C_3C_1^4}$.
\end{proof}

\section{Proofs of Propositions \ref{prop:equality} and \ref{prop:auto}}

Before the proofs, some more notation is necessary.  Let $|\lambda_1| \ge |\lambda_2| \ge \dots \ge |\lambda_N|$ denote the eigenvalues of $P$.  Denote the $f_1, \dots, f_N :V\rightarrow \mathbb{R}$ as the corresponding the eigenfunctions of $P$.

The following is a proof of Proposition \ref{prop:auto}.
\begin{proof}
Proposition 1 in \cite{khabbazian2016novel} says that
\[Cov(\tilde y(X_{p(\tau)}), \tilde y(X_\tau)) = \sum_{\ell=2}^N \langle \tilde y, f_\ell \rangle_\pi^2 \lambda_\ell.\]
Thus, from the lemma and the definition of $R$, $R$ is a convex combination of the eigenvalues
\[R = \sum_{\ell=2}^N \frac{\langle \tilde y, f_\ell \rangle_\pi^2}{var_\pi(y)} \lambda_\ell.\]
\begin{lemma} \label{lem:levin} 
\[\sum_{\ell=2}^{N}  \frac{\langle \tilde y, f_\ell\rangle_\pi^2}{var_\pi(y)} = 1.\]
\end{lemma}
The proof of Lemma \ref{lem:levin} is given on page 342 of \cite{levin2009markov} and is repeated below for completeness.

From Theorem 1 in \cite{treevar},
\begin{equation} \label{eq:var}
\sigma_{\tilde \mu}^2 \ = \  \sum_{\ell=2}^{N} \langle \tilde y, f_\ell \rangle_\pi^2 \mathbb{G}(\lambda_\ell).  
\end{equation}
Applying Jensen's inequality yields
\begin{equation}\label{eq:jensen}
\sigma_{\tilde \mu}^2 =  var_\pi(\tilde y)\sum_{\ell=2}^{N} \frac{\langle \tilde  y, f_\ell \rangle_\pi^2}{var_\pi(\tilde y)} \mathbb{G}(\lambda_\ell) \ge 
var_\pi( \tilde y) \mathbb{G}(R)
.
\end{equation}
\end{proof}

The proof of Proposition \ref{prop:equality} is similar. 
\begin{proof}
In the case when $\tilde y(i) = \mu + \sigma f(i)$, this implies that $\tilde y = \mu f_1 + \sigma f_j$.  By the orthonormality of the eigenvectors, $ \frac{\langle \tilde y, f_\ell \rangle_\pi^2}{var_\pi(\tilde y)} = \textbf{1}\{j = \ell\}$ for $\ell>1$.  As such, the inequality in equation \eqref{eq:jensen} holds with equality.
\end{proof}

The following is a proof of Lemma \ref{lem:levin}
\begin{proof} 
\[var_\pi (\tilde y) =  E_\pi(\tilde y^2) - (E_\pi \tilde y)^2 = \sum_{j=1}^{N}  \langle \tilde y, f_j\rangle_\pi^2 - (E_\pi \tilde y)^2 = \sum_{j=2}^{N}  \langle \tilde y, f_j\rangle_\pi^2.\]
\end{proof}

Proposition \ref{prop:auto} presumes that $\mathbb{G}$ is convex.  Figure \ref{fig:G} plots $\mathbb{G}$ for twenty different Galton-Watson trees with offspring distribution $p(0) = .1, p(1) = .1, p(2) = .3, p(3) = .5$.  This offspring  distribution has expected value 2.2.  The construction of each tree was stopped when it reached 5000 nodes; if it failed to reach 5000 nodes, then the process was started over.  In these simulations and in others not shown, $\mathbb{G}$ is often convex.  When it is not convex, its second derivative is positive away from $-1$.  This simulation was selected because it shows that even when the trees are sampled from the same distribution, even when there is nothing strange about the offspring distribution (e.g. all moments are finite), even when it is a very big tree, even under all of these nice conditions some of the trees have a convex $\mathbb{G}$ and some of the trees have a non-convex $\mathbb{G}$.  Similar results hold when the trees have 500 nodes; the only thing that changes is that the red regions extend slightly further away from $-1$.

\begin{figure}[h]
\centering
\includegraphics[width=\textwidth]{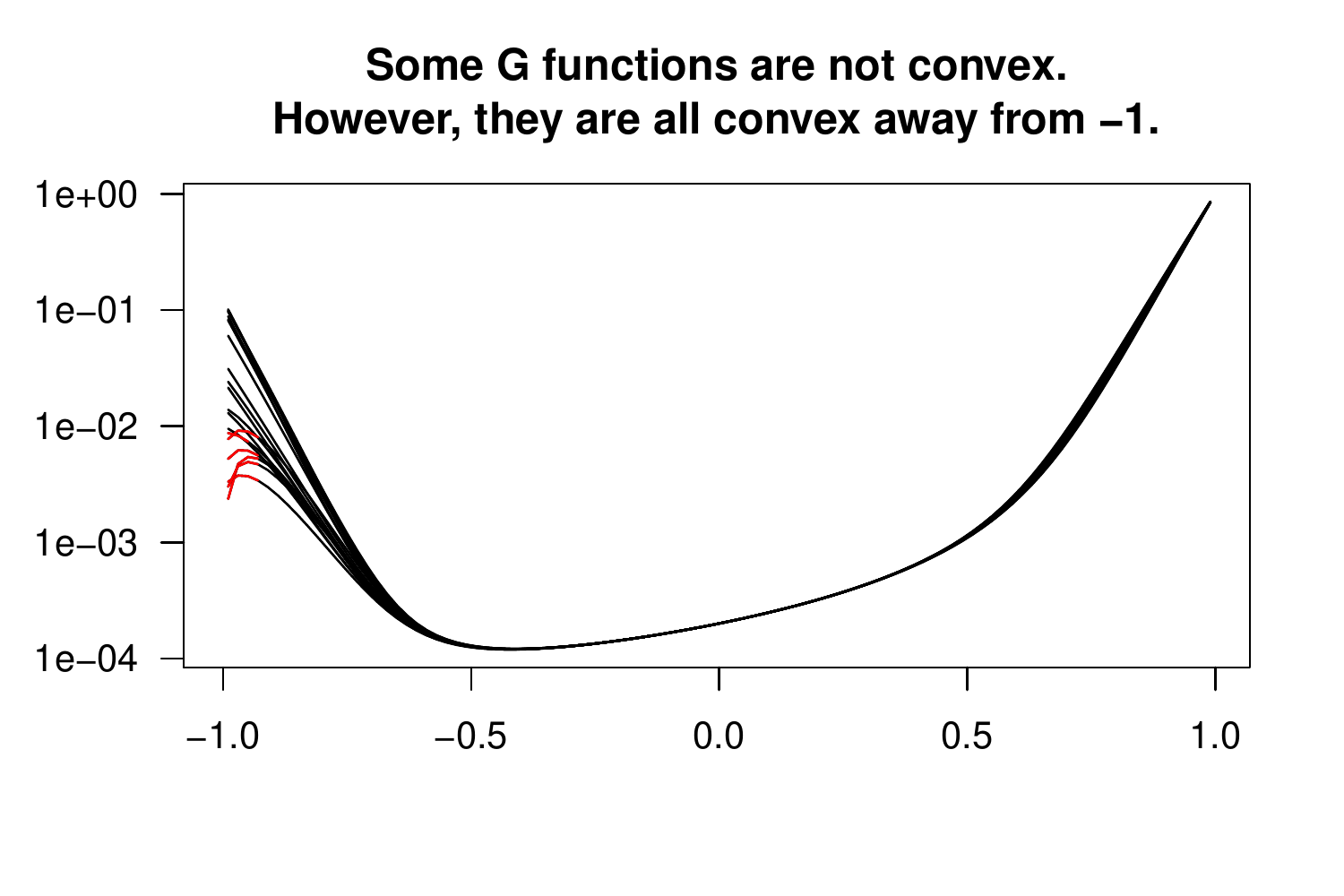}
\caption{Each line corresponds to a function $\mathbb{G}$ for a random Galton-Watson tree.  Some of the curves are not convex; the regions of these functions which have a negative second derivative are highlighted in red. While some of the black lines appear to have a negative second derivative, they do not;  this illusion is due to the log transformation on the vertical axis.}
\label{fig:G}
\end{figure}

%
%
%
%
%
%
%

\bibliographystyle{apalike}
\bibliography{TV-references}

\begin{thebibliography}{}

\bibitem[WHO, 2013]{WHO}
 (2013).
\newblock {\em Introduction To HIV/AIDS And Sexually Transmitted Infection
  Surveillance Module 4: Introduction to Respondent-drive Sampling.}
\newblock World Health Organization \& UNAIDS.
\newblock http://applications.emro.who.int/dsaf/EMRPUB\_2013\_EN\_1539.pdf.

\bibitem[Benjamini and Peres, 1994]{benjamini1994markov}
Benjamini, I. and Peres, Y. (1994).
\newblock Markov chains indexed by trees.
\newblock {\em The Annals of Probability}, pages 219--243.

\bibitem[Biemer and Christ, 2007]{biemer2007weights}
Biemer, P.~P. and Christ, S. (2007).
\newblock Weighting survey data.
\newblock {\em International Handbook of Survey Methodology}.

\bibitem[Bollen et~al., 2016]{kenneth2016weights}
Bollen, K.~A., Biemer, P.~P., Karr, A.~F., Tueller, S., and Berzofsky, M.~E.
  (2016).
\newblock Are survey weights needed? a review of diagnostic tests in regression
  analysis.
\newblock {\em The Annual Review of Statistics and Its Application},
  3(1):375--392.

\bibitem[Clauset et~al., 2009]{clauset2009power}
Clauset, A., Shalizi, C.~R., and Newman, M.~E. (2009).
\newblock Power-law distributions in empirical data.
\newblock {\em SIAM review}, 51(4):661--703.

\bibitem[Durrett, 2010]{durrett}
Durrett, R. (2010).
\newblock {\em Probability: theory and examples. Edition 4.1}.
\newblock Cambridge university press.

\bibitem[Gile, 2011]{gile2011improved}
Gile, K.~J. (2011).
\newblock Improved inference for respondent-driven sampling data with
  application to hiv prevalence estimation.
\newblock {\em Journal of the American Statistical Association}, 106(493).

\bibitem[Goel and Salganik, 2009]{goel2009respondent}
Goel, S. and Salganik, M.~J. (2009).
\newblock Respondent-driven sampling as markov chain monte carlo.
\newblock {\em Statistics in medicine}, 28(17):2202--2229.

\bibitem[Goel and Salganik, 2010]{goel2010assessing}
Goel, S. and Salganik, M.~J. (2010).
\newblock Assessing respondent-driven sampling.
\newblock {\em Proceedings of the National Academy of Sciences},
  107(15):6743--6747.

\bibitem[Handcock et~al., 2016]{RDSpackage}
Handcock, M.~S., Fellows, I.~E., and Gile, K.~J. (2016).
\newblock {\em RDS: Respondent-Driven Sampling}.
\newblock Los Angeles, CA.
\newblock R package version 0.7-5.

\bibitem[Heckathorn, 1997]{heckathorn1997respondent}
Heckathorn, D.~D. (1997).
\newblock Respondent-driven sampling: a new approach to the study of hidden
  populations.
\newblock {\em Social problems}, pages 174--199.

\bibitem[Holland et~al., 1983]{holland1983stochastic}
Holland, P., Laskey, K., and Leinhardt, S. (1983).
\newblock Stochastic blockmodels: First steps.
\newblock {\em Social Networks}, 5(2):109--137.

\bibitem[Jones et~al., 2004]{jones2004markov}
Jones, G.~L. et~al. (2004).
\newblock On the markov chain central limit theorem.
\newblock {\em Probability surveys}, 1(299-320):5--1.

\bibitem[Khabbazian et~al., 2016]{khabbazian2016novel}
Khabbazian, M., Hanlon, B., Russek, Z., and Rohe, K. (2016).
\newblock Novel sampling design for respondent-driven sampling.
\newblock {\em arXiv preprint arXiv:1606.00387}.

\bibitem[Levin et~al., 2009]{levin2009markov}
Levin, D.~A., Peres, Y., and Wilmer, E.~L. (2009).
\newblock {\em Markov chains and mixing times}.
\newblock American Mathematical Soc.

\bibitem[Lorrain and White, 1971]{lorrain1971structural}
Lorrain, F. and White, H.~C. (1971).
\newblock Structural equivalence of individuals in social networks.
\newblock {\em The Journal of mathematical sociology}, 1(1):49--80.

\bibitem[Lu et~al., 2012]{lu2012sensitivity}
Lu, X., Bengtsson, L., Britton, T., Camitz, M., Kim, B.~J., Thorson, A., and
  Liljeros, F. (2012).
\newblock The sensitivity of respondent-driven sampling.
\newblock {\em Journal of the Royal Statistical Society: Series A (Statistics
  in Society)}, 175(1):191--216.

\bibitem[Pfeffermann, 1996]{pfeffermann1996weights}
Pfeffermann, D. (1996).
\newblock The use of sampling weights for survey data analysis.
\newblock {\em Statistical Methods in Medical Research}, 5(3):239--61.

\bibitem[Rohe, 2015]{treevar}
Rohe, K. (2015).
\newblock Network driven sampling; a critical threshold for design effects.
\newblock {\em arXiv preprint arXiv:1505.05461}.

\bibitem[Rohe et~al., 2011]{rohe2012sp}
Rohe, K., Yu, B., and Chatterjee, S. (2011).
\newblock Spectral clustering and the high dimensional stochastic blockmodel.
\newblock {\em The Annals of Statistics}, 39(4):1878--1915.

\bibitem[Salganik, 2006]{salg}
Salganik, M.~J. (2006).
\newblock Variance estimation, design effects, and sample size calculations for
  respondent-driven sampling.
\newblock {\em Journal of Urban Health}, 83(1):98--112.

\bibitem[Salganik and Heckathorn, 2004]{salganik2004sampling}
Salganik, M.~J. and Heckathorn, D.~D. (2004).
\newblock Sampling and estimation in hidden populations using respondent-driven
  sampling.
\newblock {\em Sociological methodology}, 34(1):193--240.

\bibitem[Strogatz, 2001]{strogatz2001exploring}
Strogatz, S.~H. (2001).
\newblock Exploring complex networks.
\newblock {\em Nature}, 410(6825):268--276.

\bibitem[Valliant et~al., 2013]{valliant2013weights}
Valliant, R., Dever, J.~A., and Kreuter, F. (2013).
\newblock {\em Practical Tools for Designing and Weighting Survey Samples}.
\newblock Springer New York.

\bibitem[Volz and Heckathorn, 2008]{volz2008probability}
Volz, E. and Heckathorn, D.~D. (2008).
\newblock Probability based estimation theory for respondent driven sampling.
\newblock {\em Journal of Official Statistics}, 24(1):79.

\bibitem[White et~al., 2015]{white2015strengthening}
White, R.~G., Hakim, A.~J., Salganik, M.~J., Spiller, M.~W., Johnston, L.~G.,
  Kerr, L., Kendall, C., Drake, A., Wilson, D., Orroth, K., et~al. (2015).
\newblock Strengthening the reporting of observational studies in epidemiology
  for respondent-driven sampling studies:?strobe-rds? statement.
\newblock {\em Journal of clinical epidemiology}, 68(12):1463--1471.

\end{thebibliography}

\end{document}